%% file: arxiv_version.tex
\documentclass[runningheads,a4paper,11pt]{llncs}
\usepackage[margin=1.1in]{geometry}
\usepackage{amsmath,amssymb,graphicx,xcolor,pgfplots,enumitem,comment,newtxtext,newtxmath,subcaption,url,doi,csquotes,hyperref,cleveref}

\usepackage{amsthm}
\usepgfplotslibrary{dateplot}
\pgfplotsset{compat=newest}

\newcommand{\g}{g_\gamma\(u\)}
\newcommand{\qt}{\enquote}

\DeclareMathOperator*{\argmax}{arg\,max}
\renewcommand{\(}{\left(}
\renewcommand{\)}{\right)}
\newcommand{\lt}{\left[}
\newcommand{\rt}{\right]}

\newtheorem*{theorem*}{Theorem}
\Crefname{notat}{Notation}{Notations}
\newtheorem{notat}{Notation}
\newtheorem{notation}[notat]{Notation}

\begin{document}
\title{Catastrophe by Design in Population Games:\\ Destabilizing Wasteful Locked-in Technologies
\thanks{Stefanos Leonardos and Georgios Piliouras gratefully acknowledge MOE AcRF Tier 2 Grant 2016-T2-1-170. Georgios Piliouras also acknowledges grant PIE-SGP-AI-2018-01, NRF2019-NRF-ANR095 ALIAS grant and NRF 2018 Fellowship NRF-NRFF2018-07.}}
\titlerunning{Catastrophe by Design in Population Games}

\author{Stefanos Leonardos\inst{1} \and 
Iosif Sakos\inst{1} \and 
Costas Courcoubetis\inst{1} \and 
Georgios Piliouras\inst{1}}
\authorrunning{S. Leonardos et. al}

\institute{Singapore University of Technology and Design, 8 Somapah Rd, Singapore 487372, Singapore\\
\email{\{stefanos\_leonardos,costas,georgios\}@sutd.edu.sg}, \email{iosif\_sakos@mymail.sutd.edu.sg}}

\maketitle
\begin{abstract}
In multi-agent environments in which coordination is desirable, the history of play often causes lock-in at sub-optimal outcomes. Notoriously, technologies with a significant environmental footprint or high social cost persist despite the successful development of more environmentally friendly and/or socially efficient alternatives. The displacement of the status quo is hindered by entrenched economic interests and network effects. To exacerbate matters, the standard mechanism design approaches based on centralized authorities with the capacity to use preferential subsidies to effectively dictate system outcomes are not always applicable to modern decentralized economies. What other types of mechanisms are feasible? \par
In this paper, we develop and analyze a mechanism that induces transitions from inefficient lock-ins to superior alternatives. This mechanism does not exogenously favor one option over another -- instead, the phase transition emerges endogenously via a standard evolutionary learning model, Q-learning, where agents trade-off exploration and exploitation. Exerting the same transient influence to both the efficient and inefficient technologies encourages exploration and results in irreversible phase transitions and permanent stabilization of the efficient one. On a technical level, our work is based on bifurcation and catastrophe theory, a branch of mathematics that deals with changes in the number and stability properties of equilibria. Critically, our analysis is shown to be structurally robust to significant and even adversarially chosen perturbations to the parameters of both our game and our behavioral model.

\end{abstract}

\input{intro}

\input{model}
\input{protocols}

\input{analysis}

\input{conclusions}

\section*{Conflict of interest} The authors declare that they have no conflict of interest.

\bibliographystyle{splncs04}
\bibliography{bifurcation_bib,anor_bifurcation_bib,wine_ec_sagt_bib}  

\appendix
\input{robustness}
\input{applications}
\input{omitted_proofs}

\end{document}

%% file: intro.tex
\section{Introduction}\label{sec:introduction}

Efficient and sustainable innovations often fail to reach the critical mass of adoption that is needed to supersede socially harmful and energy-wasteful status-quo technologies \cite{Rob67,Rog03,Cho10} due to positive network effects and lock-in costs \cite{Mak10}. This frequently observed deadlock, termed \emph{inefficient lock-in} \cite{Mas19}, is a recurrent theme in the economics literature. Recent (technological) examples include the transportation sector --- where the existing infrastructure (network effect) for CO\textsubscript{2} emitting vehicles hinders or, at least, slows down the adoption of greener alternatives (e.g., platooning systems, electric cars or aircrafts powered by renewable energy) --- updates in the Internet's infrastructure --- e.g., the long-anticipated transition from IPV4 to IPV6, a problem also known as \emph{protocol ossification}, and, notoriously, Proof of Work mining, a setting that has attracted widespread attention by academics, entrepreneurs and policymakers \cite{Fia19,Gor19}.\footnote{Manifestations of this phenomenon are prevalent across the whole spectrum of social, economic, political, and natural sciences: persistence of bad social norms \cite{Sme20}, inefficient business practices and procedures \cite{Sal95,Mas19}, fierce market competition over supply and demand of products and services (also known as \qt{life-cycle competition}) \cite{Far07,Chi14,har16}, community or committee voting in managerial boards \cite{Ang07,Mas17}, mood disorders \cite{Lee14}, investments in undeveloped regions with tourism potential \cite{Mak10} are only some of the numerous settings in which this problem is observed.}\par
While this issue is not entirely new, \cite{Art89,Cow90,Sha99}, both experimental and theoretical studies suggest that little is known about how to move out of a lock-in \cite{Dev03,Lee07,Kes12}. Surprisingly, in all these cases, the challenge of transitioning to an existing alternative is primarily not \emph{technological} but rather \emph{game-theoretical} in nature. Yet, game theory alone can only be used to understand whether a lock-in will occur: since most standard solution concepts are invariant to the history of play, they cannot offer much insight on whether such a lock-in can be overcome unless they are coupled with a proper dynamic, behavioral model \cite{Bra06}. In this respect, the need to develop a suitable behavioral model that will combine adaptive learning with strategic behavior and which will be suitable to make predictions in games with a history of lock-in has been receiving increasing attention \cite{Mas19}.\par
In this paper, we address this problem from the perspective of mechanism design and focus on developing an economically feasible mechanism that will provably allow a central planner, or more generally, an external authority with partial and only temporary influence over agents' behavior, to move the system out of the inefficient lock-in and lead it to a socially desirable outcome. \par
From a theoretical perspective, we stylize this situation as a coordination or evolutionary game with multiple equilibria \cite{Far07}. Accordingly, the population state (or equilibrium) in which the currently prevailing, yet wasteful technology is used by all agents is evolutionary stable and small perturbations, i.e., adopters of the alternative technology, are doomed to fail \cite{Meh17}. In turn, the externality that agents' actions impose on the society is not quantified in their short term utility function and hence, does not shape their current decisions. This creates a deadlock: a situation --- among many known in social and economic sciences --- in which selfish behavior stands at odds with the social good. Yet, mechanism design, the natural approach to attack this problem, is powerless in this setting. The inherently decentralized nature of contemporary economies and social interactions raises challenges that severely lessen the applicability of off-the-shelf solutions. The reasons are well-rooted in the structure of large population networks (that represent today's transborder societies) and the behavioral patterns of the participating agents. \par
Specifically, mechanisms that consistently subsidize socially beneficial behavior are not economically feasible and would be subject to gaming in the long run. Even more of a showstopper is the fact that many top-down policies that treat differentially one option versus another have been shown to be rather hard, if not outright impossible, to sustain in practice or have frequently failed to offset users' potential losses from prevailing network effects. More importantly, standard expected utility models, the bedrock of classic mechanism design, are arguably too simplistic to model agents' behavior in practice for several reasons: hedging against volatility or extreme events, risk attitudes, collusions, politics/governance, to name but a few. It would thus be important to develop solutions that are robust to more complex behavioral assumptions. 

\paragraph{Our model:} To address this problem, we introduce an evolutionary game-theoretic model that captures agent behavior in decentralized systems with network effects. The agents' strategic interaction is stylized as a \emph{population game} whose state updates are governed by a \emph{revision protocol} that balances exploitation (expected utility maximization) and exploration of suboptimal alternatives (hedging) \cite{Joh18}. These two elements generate an evolutionary dynamic that can then be used to reason about agents' decision making in this setting \cite{San10}. \par
The advantage of having both a game-theoretic model (Section \ref{sec:model}) as well as learning-theoretic model (Section \ref{sec:behavior}) is that it allows us to formally argue about the stability of the equilibria of the game. For example, in the simplest possible game-theoretic model of competition between a prevalent, yet wasteful technology $\(W\)$ and a more efficient (both socially and individually) alternative $\(S\)$, we have that the utility of using the $W$ (respectively $S$) strategy increases linearly in the number of other agents that are using the same technology. This results in three types of fixed points: everyone using $W$, everyone using $S$, and a \qt{mixed} population case at the exact split where both technologies are equally desirable/profitable. Intuitively, this mixed state is an unstable equilibrium as a slight increase of the fraction of $W$ (respectively $S$) adopters is enough to break ties and encourage convergence to a monomorphic state. However, to make this discussion concrete, we need to formally describe how a mixed population state (i.e., the $W/S$ split) evolves over time. \par
To model the adaptive behavior of the agents, we use the Boltzmann Q-learning dynamics \cite{Wat92,Tan97} (one of the most well-known models of evolutionary reinforcement learning that can also be derived as a limit of the well known Experience-Weighted Attraction behavioral game theory model \cite{Cam99,Cam03,Gal13}). The decision of each agent, or equivalently of each unit of investment, is whether to adopt technology $W$ or technology $S$ given that $x\in[0,1]$ fraction of the population adopts technology $W$. According to the Q-learning behavioral model, the agents update their actions by keeping track of the collective past performance --- in particular, of a properly defined Q-score --- of their available strategies ($W$ and $S$). Roughly (see \Cref{sec:model} for the formal specifications and notation), if we denote with $u\(W,x\)$ and $u\(S,x\)$ the utility of an investment unit from either of the two technologies $W$ and $S$ when the state of the population is $\(x,1-x\)$, with $x\in[0,1]$ denoting the fraction of $W$ investments, then the Q-learning dynamics are given by the following scheme
\begin{equation}\label{eq:intro}
\dot x=x[\underbrace{u\(W,x\)-u\(S,x\)}_{\text{Replicator Dynamics}}-T\cdot\underbrace{\(x\ln{x}+\(1-x\)\ln{\(1-x\)}\)}_{\text{Entropy}}]\,.
\end{equation}
In a far-reaching twist, each agent's utility function is enhanced by an entropy term that is weighted by a parameter $T$, termed \emph{temperature} or \emph{rationality}. When $T=0$, the dynamics are precisely the replicator dynamics \cite{San10,Pan16,Boo19} and they recover the Nash equilibria of the game, whereas large values of $T$ lead to uniform randomization between the available options. Informally, low temperatures capture cool-headed agents that focus primarily on strategies with good historical performance, whereas high temperatures favor hedging and exploration. \par
This interpretation suggests that $T$ is a behavioral attribute of the population which cannot be influenced by an external authority. However, from equation \eqref{eq:intro}, it can be seen that the application of a \emph{tax-rate} (or any other mechanism) to rescale agents' utilities, is essentially equivalent to a change in parameter $T$, i.e., to a change in the agents' rationality \cite{Wol12,Yan17}. This is precisely the intuition that we exploit here to prompt the transition from the \emph{evolutionary stable}, yet inefficient, equilibrium of the locked-in technology $W$ (cf. \Cref{prop:evolutionary}) to the adoption of the desirable technology $S$.\medskip

\paragraph{Our solution: Catastrophe Design.} 
Our main contribution is to formally prove that there exists a simple, robust, and transient catastrophe-based mechanism to destabilize the $W$ equilibrium and enforce the $S$ equilibrium. The proposed mechanism, while simple to implement, has provable 
guarantees that critically exploit the complex underlying geometric structure of the Q-learning dynamics, cf. \Cref{sec:results}. As a first step in the process, we provide a complete characterization of both the number as well as the stability of the equilibria of the Q-learning dynamics, also known as Quantal Response Equilibria (QRE) \cite{Mck95}, given game-theoretic models of $W/S$ competition for all values of $T$ (Theorem \ref{thm:main}). This is a question of independent interest as it explores the possible limit system behaviors in the lack of any controlling mechanism. \par
Then, we describe how by simply raising taxation/temperature up to (slightly beyond) a critical value and then reducing it down to zero results into convergence to the socially optimal $S$ equilibrium (Theorem \ref{cor:convergence}). The mechanism is formally described in \Cref{sub:hysteresis}. Here, we provide an informal statement for the reader's convenience.

\begin{theorem*}[Informal statement of \Cref{cor:convergence}]
For any initial population state, there exists a finite sequence of control variable levels $T = \langle T_0=0, T_1, \ldots, T_n=0\rangle$ such that through the following procedure:
\begin{itemize}
\item scale the control variable to $T_i$, and
\item wait until the system converges to a QRE
\end{itemize}
the system is going to converge to the desirable state $x = 0$, which corresponds to the energy-friendly technology $S$. In particular, the sequence of control levels starts from $T_0=0$ (no control), increases temporarily above a critical level and is reset back to $T_n=0$.
\end{theorem*}

Our approach is based on the combination of two observations which can be visualized in \Cref{fig:introduction_plot_2} (cf. \Cref{fig:main}): 
First, the number and stability of equilibria (QRE) of the Q-learning dynamics are a function of $T$, i.e., of the trade-off level between exploration and exploitation. For $T=0$, there exist three QRE, which are precisely the three described Nash equilibria. For slightly larger $T>0$, we still have three QRE, but now due to the exploration term, all three lie in the interior of the interval $(0,1)$. Finally, beyond some critical level of the control parameter, the number of QRE drops from three to one. In particular, at this critical level, the prevailing equilibrium fuses with the unstable mixed equilibrium and at the very next moment, the now merged equilibria, cancel each other out at a phenomenon known as \emph{saddle-node bifurcation} or \emph{catastrophe} \cite{Str00,Kuz04}. \par
The second observation is that we can effectively control parameter $T$ (e.g., increase it by a multiplicative scale $\alpha$), since it is mathematically equivalent to scaling down the utility of the agents by the same factor $\alpha$ (up to time reparameterization in equation \ref{eq:intro}). In policy terms, this can be interpreted as taxation of income (i.e., a multiplicative decrease of payoffs for all actions), which results in agent behavior that is less stringent about maximizing earnings at all costs. Informally and taking this idea to its logical extreme, a taxation level of $99\%$  (i.e., a very large $T$) would effectively render the agents indifferent about payoffs and make them choose actions at random.\par

\begin{figure}[!t]
\centering
\includegraphics[scale=0.58]{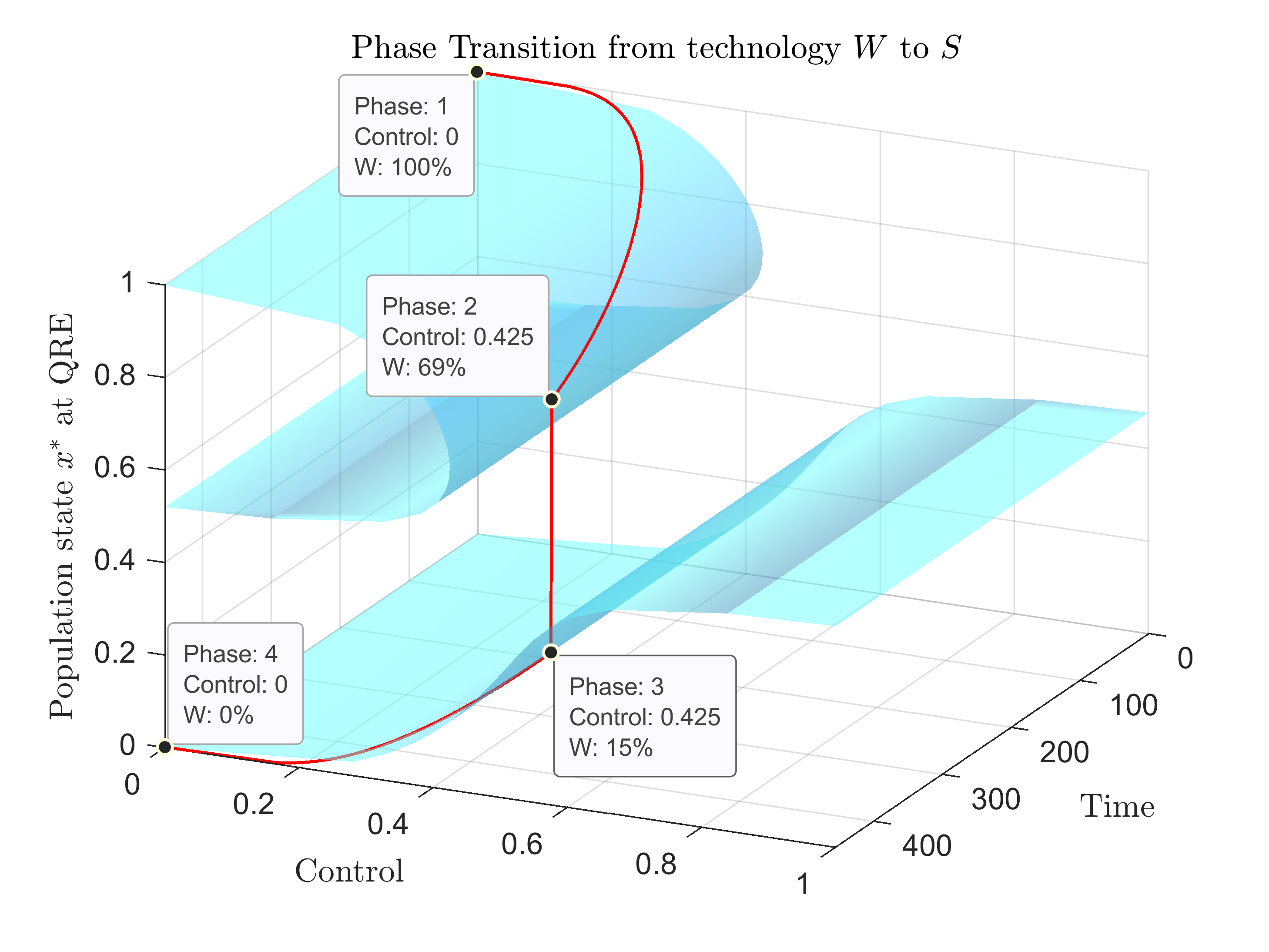}
\caption{Visualization of the QRE surface projected in time (right horizontal axis) and the transition from the prevailing wasteful technology $\(W\)$ to the efficient technology $\(S\)$ via an induced catastrophe. At phase 1, all resources are invested in $W$, and the control parameter is at $0$. As the control parameter increases, the population moves along the red line on the QRE surface. At phase 2, the control parameter reaches the critical value or tipping point at which the two upper QRE merge into one and at the very next moment, phase 3, at which the parameter $T$ is increased slightly above the critical level, the system undergoes an abrupt transition at which the upper QRE vanishes. Between these two successive time points, the population state changes from $69\%$ of investment in $W$ right before the \emph{catastrophe} to only $15\%$ immediately after. After this point, the control parameter is reset to $0$, and due to the resulting \emph{hysteresis effect} \cite{Rom15}, the population converges to the new equilibrium which corresponds to the adoption of the (new) technology $S$.}
\label{fig:introduction_plot_2}\vspace{-0.3cm}
\end{figure}

Putting these two observations together, by controlling $T$, we can control the resulting QRE and, thus, the resulting state of the system over time. More critically, when exceeding the critical level of the control parameter (\emph{tipping point}), a phase transition or catastrophe occurs at which the state behavior changes abruptly (phases $2$ and $3$ in \Cref{fig:introduction_plot_2}). This is the critical step that disrupts the prevalent outcome and encodes a new memory to the system that moves the population into the attracting region of the desirable equilibrium. At this point, the level of adoption of the new technology, or \emph{critical mass}, is typically well below the simple majority (50\%). Finally, we leverage (in principle at least) this saddle-node bifurcation (which eliminates the inefficient equilibria) to create irreversible phase transitions such that even when the controlling parameter $T$ returns to its initial state $T=0$, the state of the system is not the original undesirable stable state ($x=1$), but the target $S$ state, $x=0$ (phase 4 in \Cref{fig:introduction_plot_2}).

\paragraph{Robustness of findings.}
Critically, we stress-test our findings and establish that they are robust to modeling uncertainty/misspecifications across different axes, cf. \Cref{sec:robustness}. In \Cref{sub:robust_e}, we allow (possibly adversarial) uncertainty/perturbations in the dynamics and the population game-theoretical model. For the reader's convenience, our main result is restated in the following Theorem.

\begin{theorem*}[Informal statement of \Cref{thm:stability_e}]
Under state-dependent, bounded perturbations in the evolution of the Q-learning dynamics, the proposed catastrophe mechanism yields a transition from the inefficient to the efficient equilibrium with the only difference that the population stabilizes at some neighborhood of the QRE of the unperturbed system in the intermediate phases. 
\end{theorem*}
In \Cref{sub:robust_a}, we explore utility functions that capture strong superadditive effects on network valuation. These models introduce significant difficulties as lie outside the typical framework of evolutionary game theory with multi-linear utility functions. Again, our main result is informally stated below.

\begin{theorem*}[Informal statement of \Cref{thm:unique}]
For positive network effects of arbitrary intensity, i.e., for arbitrary non-linear payoff functions that yield superadditive value in the degree of adoption, the proposed catastrophe mechanism yields a transition from the inefficient to the efficient equilibrium.
\end{theorem*}

The findings of the robustness analysis justify the selection of our baseline modeling assumptions with an eye towards simplicity primarily for expositional purposes. In short, the main takeaway from the two robustness theorems is that \emph{our results can be directly extended for a wide range of modeling designs and approaches}. Technically, while the number of equilibria and the stability properties of the dynamics may change under these perturbed conditions, the resulting control mechanism can be applied without any significant changes. 

\paragraph{Other related work.}
While applications of catastrophe theory in the field of economics have received mixed reactions in the past, not least due to their original hype \cite{Ros07}, the formal connections between bifurcations, hysteresis effects, and coordination games have started to gain increasing attention \cite{Tuy03,Kia12,Rom15}. The idea of leveraging these complex dynamics in the arsenal of mechanism design has been only recently proposed \cite{Wol12,Yan17}, yet its formal treatment remains unexplored. Applying these concepts in a real options setting and bringing forth theoretically provable properties concerning a market's network effects and the adoption rate of competing technologies, our work extends previous studies that identify and reason about these connections mainly via experiments \cite{Bec08,Mak10,Bra16,Sme20}. Further supporting the scarcity of a formal approach, \cite{Bat16} emphasizes the increasing complexity of economic and social interactions and argues about the need to join forces and develop tools from complexity theory, as a complement to existing economic modeling approaches. Focusing on physical applications, the main field where catastrophe theory is studied, \cite{Sch09} uses the fold catastrophe model to identify early warning signs of critical transitions in complex systems. In relation to mechanism design, such results can amplify the benefits of the catastrophe model as a tool to create, rather than prevent, critical transitions in the hands of policymakers who can, thus, timely predict the effects of their policies. 

\paragraph{Outline.} The rest of the paper is structured as follows. In \Cref{sec:model,sec:behavior}, we provide the formal definitions of the game-theoretical and behavior framework. Our working assumptions (which we relax in \Cref{sec:robustness}) are presented in \Cref{sub:elementary}. \Cref{sec:results} contains the technical work up to the design of the main catastrophe mechanism that is presented and explained in \Cref{sub:hysteresis}. In \Cref{sec:robustness}, we establish the robustness of the proposed mechanism under a broad set of assumptions. We conclude with an illustrative application in the context of blockchain mining in \Cref{sec:applications}. Technical proofs and accompanying materials are deferred to \Cref{app:appendix,app:b}.

%

%% file: model.tex

\section{Model: Population Game}\label{sec:model}
We consider a society or population\footnote{The theory of population games and revision protocols that we use here closely follows \cite{San10}.} $p$ of acting agents or investors who form a continuum of mass $K>0$. Here $K$ denotes the total available capital or resources that the agents are willing to invest and is expressed in monetary terms. The set of available actions or technologies is denoted by $A=\{W,S\}$, where $W$ is the costly (wasteful) technology and $S$ is the innovative (socially and individually preferable) technology. In particular, investing one monetary unit in technology $W$ incurs a cost of $\gamma>0$ to the investor, while the investment in technology $S$ incurs zero cost. The latter assumption ensures that (rational) agents may disregard potential alternatives and invest all available resources on either of the two technologies (there is no loss from doing so). Accordingly, the set of \emph{population states} is $X=\{\(x,1-x\): x\in[0,1]\}$, where $x\in[0,1]$ denotes the fraction of agents (in terms of capital or resources) in population $p$ that is choosing technology $W$. Thus, we may slightly abuse notation and refer to $x\in[0,1]$ as the \emph{population state}. \par
The \emph{payoff function}, $u:A\times X\to \mathbb R^2$, assigns to each population state a vector of payoffs, one for each strategy in $A$. We assume that the total value created by each technology $\{W,S\}$ in $A$ depends on the fraction $xK$ of capital that has been invested in that technology via a parameter $\alpha>0$ and that the total value is distributed evenly among all invested units. We assume that either technology can generate an aggregate value $V>K$, if fully adopted by the population. In particular, the values $V\(W,x\)$ and $V\(S,x\)$ created by technologies $W$  and $S$ depend on the population state $x\in [0,1]$ via the relationships 
\begin{equation}\label{eq:value}
V\(W,x\)=V\cdot\(xK\)^\alpha \qquad \text{and} \qquad V\(S,x\)=V\cdot\lt\(1-x\)K\rt^\alpha\,.
\end{equation}
Different values of $\alpha$ give rise to different \emph{network effects or externalities}.\footnote{This is discussed in more detail in \Cref{rem:alpha}.} In particular, $\alpha<1$ implies subadditive value (it is optimal for the population to split), $\alpha=1$ implies linear value, and $\alpha>1$ implies superadditive value, i.e., the population is better off if it fully adopts either of the two technologies. Combining the above, the payoff of each strategy $\{W,S\}\in A$ is given by 
\begin{subequations}
\label{eq:payoffs}
\begin{align}
u\(W,x\)&=V\(W,x\)\cdot\frac{1}{xK}-\gamma=V\cdot\frac{\(xK\)^\alpha}{xK}-\gamma=V\cdot\(xK\)^{\alpha-1}-\gamma\\
u\(S,x\)&=V\(S,x\)\cdot\frac{1}{\(1-x\)K}=V\cdot\frac{\lt\(1-x\)K\rt^\alpha}{xK}=V\cdot\lt\(1-x\)K\rt^{\alpha-1}
\end{align}
\end{subequations}
Hence, the \emph{average payoff} obtained by the members of the population at state $x\in[0,1]$ is equal to 
\begin{equation}\label{eq:aggregate}
\bar{u}\(x\)=xu\(W,x\)+\(1-x\)u\(S,x\)=VK^{\alpha-1}\lt x^\alpha-\frac{\gamma x}{VK^{\alpha-1}}+\(1-x\)^\alpha\rt
\end{equation}
and the \emph{aggregate payoff} that is achieved by the population as a whole is $u_A\(x\)=K\bar{u}\(x\)$. The cost $K\gamma x$ is paid by the population as a whole and hence, captures the negative externality (or cost) of the undesirable technology.

\subsection{Evolutionary Game and Nash Equilibria}\label{sub:elementary}
To study instances with positive network externalities (or direct network effects \cite{Mak10}) in which full adoption of one of the technologies is preferable, our main focus will be the case $\alpha\ge1$. For expositional purposes, we will restrict our attention to the case $\alpha=2$, but all arguments essentially carry over to any $\alpha>1$ (and to the trivial case, $\alpha=1$) as we show in \Cref{sub:robust_a}. Specifically, for $\alpha=2$, \cref{eq:payoffs,eq:aggregate} become linear in $x$,
\begin{equation}\label{eq:payoffs_2}
u\(W,x\)=VKx-\gamma, \qquad u\(S,x\)=VK\(1-x\)
\end{equation}
which allows for an equivalent --- yet more intuitive --- interpretation of the agents' interaction as a single population evolutionary game, cf. \cite{Hof98}. By substituting $x=1$ (action $W$ for the whole population) and $x=0$ (action $S$) in \eqref{eq:payoffs_2}, we can represent the game by the matrix
\begin{equation}\label{eq:game}
P = \;\;\;\bordermatrix{
~ & W & S \cr
W & VK-\gamma & -\gamma \cr
S & 0 & VK \cr}.
\tag{G1}
\end{equation}
Since $V>K> \gamma>0$, we may henceforth normalize $VK$ to $1$ and write $\gamma\to \gamma/VK$ with $\gamma \in \(0,1\)$. The equilibria of the resulting game are characterized next. The proof is standard and is presented for completeness in \Cref{app:appendix}.
\begin{proposition}\label{prop:evolutionary}
The payoff functions in \eqref{eq:payoffs_2} describe a single population, evolutionary game with three Nash equilibria: $x_1=0$, $x_2=\(1+\gamma\)/2$ and $x_3=1$ with average payoffs $\bar{u}\(x_1\)=1$, $\bar{u}\(x_2\)=\(1-\gamma\)/2$ and $\bar{u}\(x_3\)=1-\gamma$. The two pure equilibria, $x_1=0$ and $x_3=1$ are evolutionary stable, while the fully mixed equilibrium, $x_2$, is not. Equilibrium $x_1$ -- in which the desirable technology is fully adopted -- is strictly payoff and risk dominant.
\end{proposition}
This formulation provides a theoretical explanation of the reasons why the more efficient technology may not succeed. In particular, starting from the prevailing equilibrium of the wasteful technology, evolutionary stability implies that even if a small part of the population adopts the new technology, the system will not move away from the current equilibrium. Hence, although payoff and risk-dominated by the more efficient technology, the $W$ (or equivalently $x=1$) equilibrium persists. This creates the lock-in at the wasteful technology. To formally reason about potential mechanisms to disrupt this situation, we first need to develop a proper adaptive learning model that accurately describes the updates of the population states.

%% file: protocols.tex
\section{Behavioral Model}
\label{sec:behavior}

The adoption of new technologies is a gradual process at which concerned individuals adjust their actions --- distribution of their investments between old and new technologies --- via repeated interaction with their environment. In the presence of network effects, each agent's payoff critically depends on the constantly evolving \emph{population state} or equivalently on the fractions $x,1-x$ of the population that adopts either of the two technologies, cf. equation \eqref{eq:payoffs}. In this context, \cite{San10} provides several alternative microfoundations of the standard \emph{replicator dynamics} under the term \emph{revision protocols}. \par
To improve upon the sub-optimal outcomes that are often reached by the greedy and myopic updates of replicator dynamics \cite{Fia19j} and, importantly, to more accurately model agent's strategic deviations from expected utility maximization in risky choices under uncertainty (as in the present context), more elaborate models of adaptive learning have been proposed \cite{Cam99,Cam03}. Specifically, the effects of exploration and exploitation in coordination games with path-dependence (history of play) and network effects, are commonly examined via a smooth variant of \emph{Q-learning dynamics}.\footnote{This variant of Q-learning has been extensively studied in the related literature under various names. A (highly) non-exhaustive list includes \cite{Les05,Tuy06,Kia12} and \cite{Wol12} who study the dynamics as \emph{Boltzmann Q-learning dynamics} or simply \emph{Q-learning dynamics} --- we will follow this convention --- \cite{Alo10,Rom15} and \cite{San18} as \emph{logit-response} and \cite{Cou15,Mer16} as exponential reinforcement learning dynamics.}. Its connection with game-theoretical settings has been recently elaborated  \cite{Tuy03,Sat03,Sat05,Wol12,Kia12} and provides the building block for our subsequent analysis. 

\subsection{Population States with Q-learning Dynamics}

We consider a homogeneous population in which each agent is adapting their strategy by repeatedly interacting with their environment (rest of the population). The critical parameter that we need to track (and update) is the fraction $x\in[0,1]$ of the population that invests in the costly (and currently prevailing) technology $W$. To do this, we focus on the standpoint of a single agent (or unit of investment) and describe the system via the $Q$-learning dynamics which are based on the principle of \emph{Q-learning} \cite{Wat92,Cam99}. The connection to population games that is described here closely follows \cite{Tuy03,Kia12} and \cite{Yan17} without further reference.

\paragraph{$Q$-values:} At each time $t\ge0$, the learning agent assigns a value $Q_{t}\(j\)$ to each strategy $j\in A=\{W,S\}$ via the update rule
\begin{equation}\label{eq:qvalues}
Q_{t+1}\(j\)=Q_t\(j\)+\delta\lt u\(j, x_t\)-Q_t\(j\)\rt
\end{equation}
where $\delta>0$ is the learning rate and $u\(j,x_t\)$ is the reward from selecting strategy $j \in \{W,S\}$ (as given by equations \eqref{eq:payoffs}) when the distribution of the population is $x_t\in[0,1]$.

\paragraph{Strategies \& Population States:} Using the $Q$-values, the critical decision for each learning agent is the update of her choice distribution. To avoid suboptimal results by greedy updating, i.e., selection of the strategy with the highest $Q$-value, the agents incorporate in their decision problem an entropy term that rewards \emph{exploration} of the whole action space (both technologies). In particular, we assume that each agent selects their strategy $x_{t}\in[0,1]$ at time point $t\ge0$ as the (unique) solution of the convex optimization problem\footnote{For an explicit derivation of the objective function see \Cref{app:appendix}.}
\begin{align}\label{eq:protocol}
x_t=\argmax_{x\in\(0,1\)}{\left\{xQ_t\(W\)+\(1-x\)Q_t\(S\)-T\lt x\ln{x}+\(1-x\)\ln{\(1-x\)}\rt\right\}}\tag{S1}
\end{align}
where $T\ge0$ is the \emph{control parameter} that tunes the exploration rate. In particular, for $T=0$, the agent selects the action with the highest Q-value (pure exploitation), whereas for $T\to\infty$, the agent randomizes between the two available actions, i.e., investments in technologies $W$ and $S$ (pure exploration). The decision rule or \emph{revision protocol} \cite{San10} in (S1) yields the choice distribution
\begin{equation}\label{eq:update}
x_t=\frac{e^{Q_{t}\(W\)/T}}{e^{Q_{t}\(W\)/T}+e^{Q_{t}\(S\)/T}}
\end{equation}
which is known as the \emph{Boltzmann distribution}. With a slight abuse of notation, $x_t$ denotes both the learning agents' choice distribution and the state of the population. However, under the assumption that all agents are symmetric and that they are learning concurrently, both these notions are equivalent. 

\paragraph{Continuous-time dynamics:} The population game in \eqref{eq:game}, together with the revision protocol in \eqref{eq:protocol}, determine the evolutionary population dynamics. In particular, if we take the time interval to be infinitely small, this sequential joint learning process can be well approximated --- after rescaling the time horizon to $t\to\delta t/T$ --- by the continuous-time dynamics
\begin{equation}\label{eq:dynamics_pre}
\dot x = x\lt u\(W,x\)-\bar{u}\(x\) + T\sum_{j=W,S} x_j \ln{\(x_j/x\)} \rt
\end{equation}
(where $x_W:=x$ and $x_S:=1-x$), which is the desired expression of the dynamics in terms of population states $x\in[0,1]$ (rather than $Q$-values). 

\paragraph{Quantal Response Equilibria (QRE):} For any given $T\ge0$, the steady states of the system in \eqref{eq:dynamics_pre} are the values of $x\in\(0,1\)$ for which the expression on the right side becomes zero. As shown in \cite{Kia12}, these are precisely the 
\emph{Quantal Response Equilibria (QRE)} of the underlying population game \cite{Mck95}. Assuming the standard logit form, for any $T\ge0$, the QRE are defined as all points $x^*\(T\)\in[0,1]$ that satisfy the equation 
\begin{equation}\label{eq:quantal}
x^*\(T\) =\frac{e^{u\(W,x^*\)/T}}{e^{u\(W,x^*\)/T}+e^{u\(S,x^*\)/T}}.
\end{equation} 
Importantly, as shown in \cite{Yan17}, in coordination games with 2 strategies, starting from any interior point\footnote{The introduction of the exploitation term renders the choices $x=0$ and $x=1$ not admissible, since they are not in the domain of $\ln{\(x/\(1-x\)\)}$ for any $T>0$. However, for practical applications, this is a realistic assumption.} $x\in\(0,1\)$, the $Q$-learning dynamics converge to interior rest points (QRE) for any $T\ge0$. 

%% file: analysis.tex
\section{Analysis: Steady Population States} 
\label{sec:results}

Our main task is to understand how an external designer can influence the state of the population, i.e., the distribution of the aggregate investment capital across the two different technologies, by modifying the control parameter $T\ge0$. Before formally discussing the proposed mechanism in \Cref{sub:hysteresis}, we first need to reason about the solutions (steady states) and stability properties of the dynamics in equation \eqref{eq:dynamics_pre} for all different values of the control parameter $T\ge0$. We start with an observation that immediately follows from \eqref{eq:payoffs} and \eqref{eq:dynamics_pre}

\begin{lemma}\label{lem:explicit}
For a given cost parameter $\gamma\in\(0,1\)$, consider the game described by equation \eqref{eq:game}, and the revision protocol in \eqref{eq:protocol}. Then, the updates in the fraction of investments on the wasteful technology $W$, i.e., population state $x\in[0,1]$, are governed by the continuous-time dynamics 
\begin{equation}\label{eq:dynamics}
\dot x=x\(1-x\)\lt 2x-\(1+\gamma\)-T\ln{\(\frac{x}{1-x}\)}\rt.
\end{equation}
for any value $T\ge0$ of the control parameter.
\end{lemma}

When determining the steady states (QRE) and stability properties of equation \eqref{eq:dynamics} for different $T\ge0$, it is more intuitive to treat the instance $T=0$ separately. In this case, the system reduces to the well-known replicator dynamics, and its steady states are precisely the Nash equilibria of the respective evolutionary game in \eqref{eq:game}.\par
Concerning stability, the one-dimensional differential equation in \eqref{eq:dynamics} describes a vector field on the line and specifies the velocity vector $\dot x$ at each $x\in[0,1]$. Accordingly, a fixed point, i.e., a $x_0\in[0,1]$ such that $\dot x_0=0$, is \emph{stable} if the flow points towards it, i.e., if $\dot x>0$ for any $x<x_0$ and $\dot x<0$ for any $x>x_0$ such that $|x-x_0|<\epsilon$ for some $\epsilon>0$, i.e., the property needs to be satisfied in a (strictly positive) neighborhood around $x_0$. The neighborhood of $x_0$ for which this holds is called its \emph{attracting region} \cite{Pan16}. If such an attracting region does not exist, the fixed point is called unstable \cite{Str00}. 

\begin{proposition}\label{prop:t_zero}
For $T=0$, the steady states of the Q-learning dynamics in equation \eqref{eq:dynamics} are $x_1=0,x_2=\(1+\gamma\)/2$, and $x_3=1$. The steady states on the boundary, i.e., $x_1$ and $x_3$, are stable, whereas $x_2$ is unstable.
\end{proposition}
\begin{proof}
For $T=0$, the dynamics become $\dot x=x\(1-x\)\(2x-1-\gamma\)$ and the first claim follows trivially. Concerning their stability, observe that $\dot x<0$ for any $x\in\(0,\(1+\gamma\)/2\)$ and $\dot x>0$ for $x\in\(\(1+\gamma\)/2,1\)$. Hence, starting from any point other than $x=\(1+\gamma\)/2$, the system will converge to the boundary steady states, i.e, to $x_1=0$ for any initial starting point $x<x_2$ and to $x_3=1$ for any initial starting point $x>x_2$, which completes the proof.
\end{proof}
The steady states and their stability properties for $T=0$ are illustrated in \Cref{fig:stability_a}. To proceed with the general case, $T>0$, we restrict to interior population states, i.e., $x\in\(0,1\)$, so that $\ln{\(x/\(1-x\)\)}$ is well defined. For $x\in\(0,1\)$, the term $x\(1-x\)$ is always strictly positive and hence, the fixed points (QRE) and velocity of the dynamics in equation \eqref{eq:dynamics} are fully determined by the term $\lt 2x-\(1+\gamma\)-T\ln{\(x/\(1-x\)\)}\rt$. Hence, it will be convenient to introduce the following notation.
\begin{notation}\label{def:f}
For given cost and control parameters, $\gamma\in\(0,1\)$ and $T\ge 0$, let
\begin{equation}\label{eq:f_sign}
f\(x;T,\gamma\):=2x-\(1+\gamma\)-T\ln{\(\frac{x}{1-x}\), \qquad \text{for } x\in\(0,1\)}.
\end{equation}
Whenever obvious from the context, we will simplify notation and write $f\(x\)$.
\end{notation}
Keeping in mind that $T$ is viewed as a control parameter that may vary over time (in response to the actions of some exogenous actor), a two-dimensional visualization of the geometric locus of all points $x^*\in\(0,1\)$ such that $f\(x;T,\gamma\)=0$ --- i.e., of the QRE correspondence --- for selected fixed values of $\gamma$ and variable $T\ge0$ is given in \Cref{fig:sub1}. \Cref{fig:sub2} shows the QRE correspondence in a three-dimensional space --- where both $\gamma$ and $T$ are treated as variables --- and \Cref{fig:sub3} shows its projection on the cost-control, $\(\gamma,T\)$, parameter space. 

\begin{figure}[!htb]
\centering
\begin{subfigure}{0.33\linewidth}
\centering
\includegraphics[width=\textwidth]{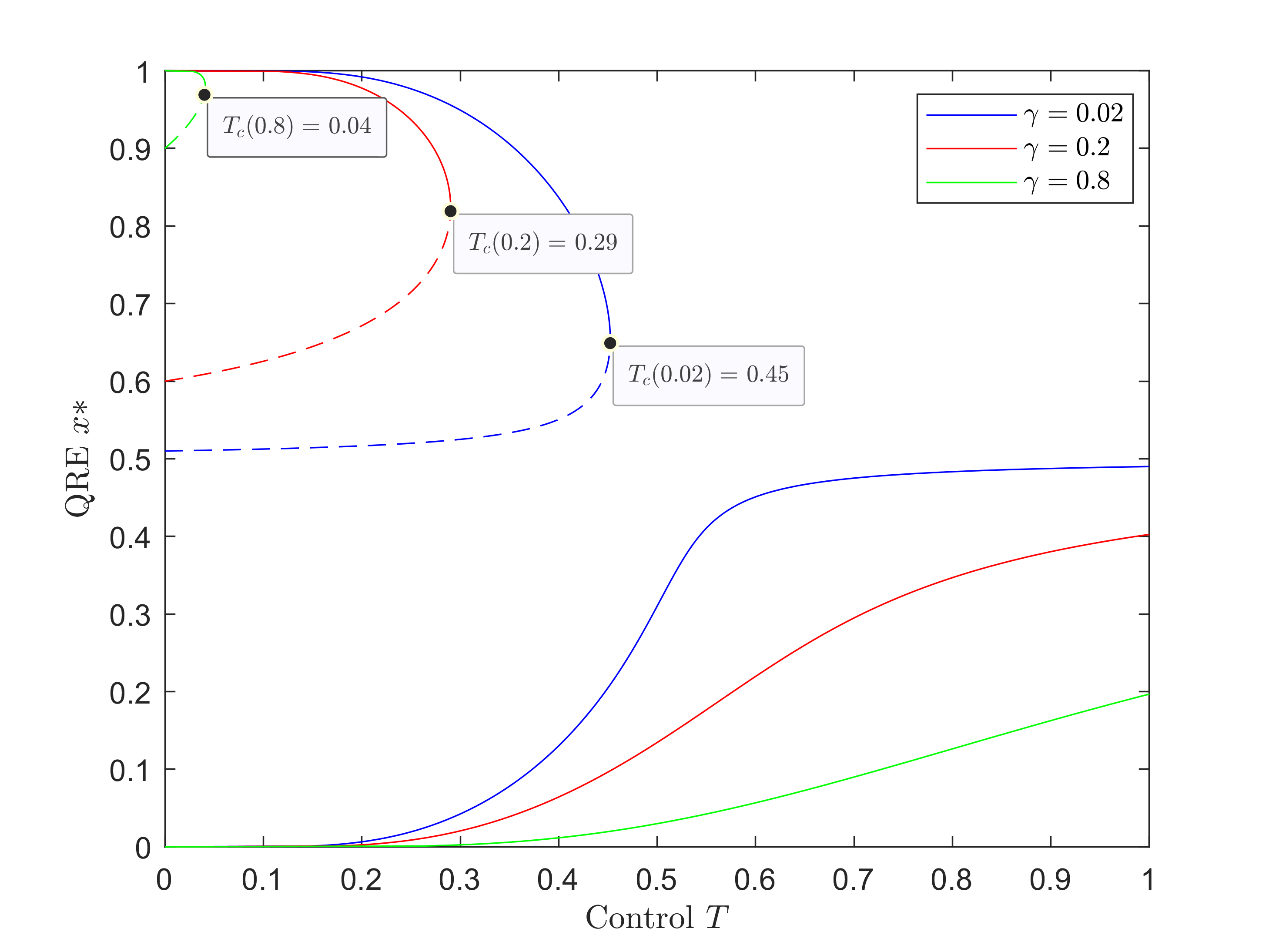}
\caption{}
\label{fig:sub1}
\end{subfigure}%
\begin{subfigure}{.33\linewidth}
\centering
\includegraphics[width=\textwidth]{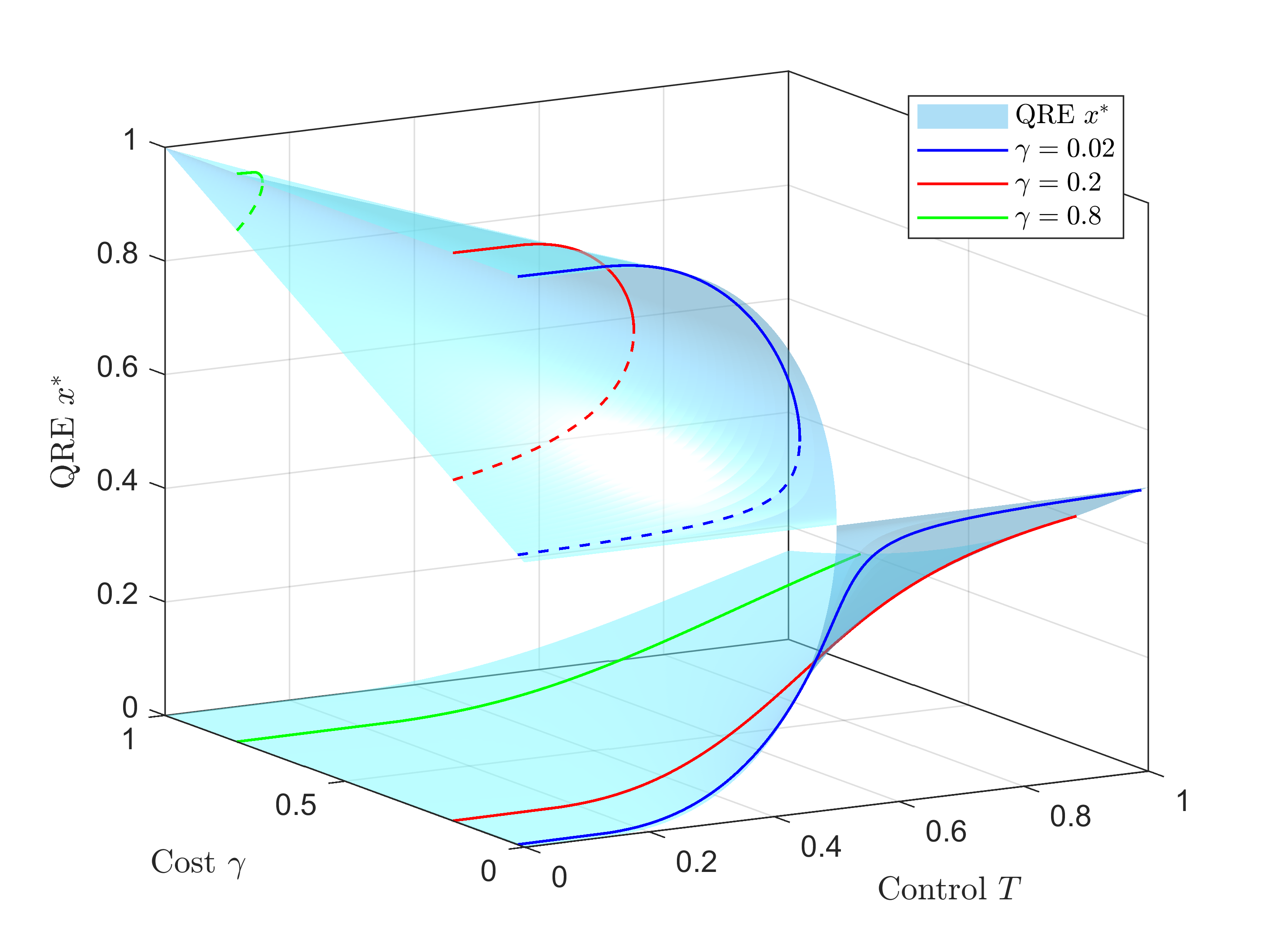}
\caption{}
\label{fig:sub2}
\end{subfigure}%
\begin{subfigure}{.33\linewidth}
\centering
\includegraphics[width=\textwidth]{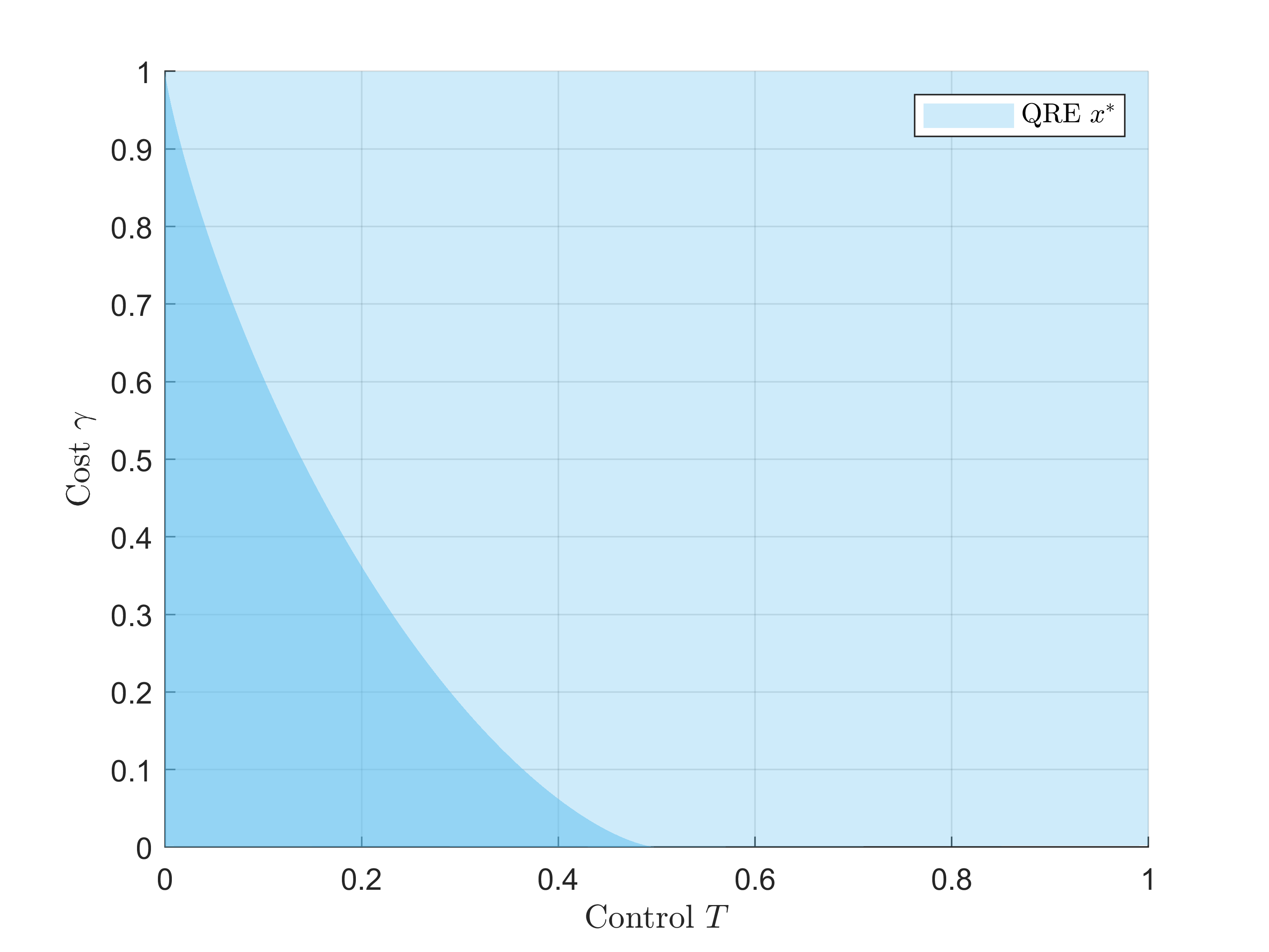}
\caption{}
\label{fig:sub3}
\end{subfigure}
\caption{\Cref{fig:sub1} shows the QRE correspondence (geometric locus of QRE) for $\gamma=0.02$ (blue line), $\gamma=0.2$ (red line), and $\gamma=0.8$ (green line). In all cases, there exists a critical value $T_c\(\gamma\)$, which determines the range of values of $T\ge0$ for which there are three, two, or one QRE. Dashed lines indicate \emph{unstable} QRE; see \Cref{thm:main}.\\
\Cref{fig:sub2} shows the QRE correspondence for all possible values of $\gamma\in\(0,1\)$. The $\(T,x^*\)$-slices at the $\gamma=0.02,0.2$, and $0.8$ levels correspond precisely to the blue, red, and green lines in \Cref{fig:sub1}. \Cref{fig:sub3} shows the projection of the QRE correspondence on the $\(\gamma,T\)$ plane. Darker areas correspond to parameter values with multiple (three) QRE. The critical values $T_c\(\gamma\)$, at which the phase transitions from one to three QRE occur, correspond to the middle boundary that separates the dark and light regions.}
\label{fig:sub}\vspace*{-0.4cm}
\end{figure}

\subsection{Critical values of the control parameter}
\label{sub:critical}
The main observation from \Cref{fig:sub} is that for each $\gamma\in\(0,1\)$, there exists a unique critical value, $T_c\(\gamma\)$ of the control parameter such that the number of steady states (QRE) depends on whether $T$ is less than, equal to or larger than $T_c\(\gamma\)$. In particular, for $T<T_c\(\gamma\)$, there exist three QRE, which for $T=0$ are precisely the Nash equilibria of the underlying game, for $T=T_c\(\gamma\)$, i.e., at the transition point, there exist two QRE, and for $T>T_c\(\gamma\)$, there remains only one QRE. In all cases, the critical temperature $T_c$ lies in the interval $\(0,1/2\)$. These properties are formalized in \Cref{thm:main}. Before proceeding with its statement, however, it will be convenient to introduce the following notation. 

\begin{notation}\label{eq:notation}
For any $T\in[0,1/2]$, let $x_{u,l}\(T\):=\frac12\(1\pm\sqrt{1-2T}\)$. Whenever obvious from the context, we will omit the dependence on $T$ and write $x_{l,u}$.
\end{notation}
Using the above, we can now formulate \Cref{thm:main}. 
\begin{theorem}[Steady states and stability of the Q-learning dynamics]\label{thm:main}
For any value of the cost parameter $\gamma\in\(0,1\)$, there exists a unique critical value of the control parameter $T_c\(\gamma\)$ with $T_c\(\gamma\)\in\(0,1/2\)$, so that the steady states of the population or equivalently, the Quantal Response Equilibria (QRE), of the continuous time Q-learning dynamics in \eqref{eq:dynamics}
\[\dot x=x\(1-x\)\lt 2x-\(1+\gamma\)-T\ln\(\frac{x}{1-x}\)\rt, \qquad x\in\(0,1\),\]
and their stability properties are determined by the relative value of $T$ in comparison to $T_c\(\gamma\)$ as follows
\begin{itemize}[topsep=0pt]
\item for $T<T_c\(\gamma\)$ there are 3 steady states $x_1,x_2,x_3$, with $x_1\in \(0,x_l\)$, $x_2\in \(\(1+\gamma\)/2,x_u\)$ and $x_3\in \(x_u,1\)$, 
\item for $T=T_c\(\gamma\)$ there are 2 steady states $x_1,x_2$, with $x_1\in\(0,x_l\)$, and $x_2=x_u$, and
\item for $T>T_c\(\gamma\)$ there is  1 steady state $x_1$, with $x_1\in\(0,x_l\)$ when $T< 1/2$ and $x_1\in \(0,1/2\)$ when $T\ge1/2$,
\end{itemize}
where $x_{l,u}$ are given in \Cref{eq:notation} whenever they exist, with $x_l<1/2$ and $x_u>\(1+\gamma\)/2$. In all cases, the steady state $x_1\in\(0,x_l\)$ is stable. For $T<T_c\(\gamma\)$, steady states $x_1\in\(0,x_l\)$ and $x_3\in \(x_u,1\)$ are stable whereas $x_2$ is not. In particular, for $T=T_c\(\gamma\)$, steady state $x_2=x_u$ is unstable.
\end{theorem}
The proof of \Cref{thm:main} relies on \Cref{lem:critical,lem:f}, for all of which the reader is deferred to \Cref{app:omitted}. The case $T=0$ which was separately treated in \Cref{prop:t_zero}, can be derived as a special case of the part $T<T_c\(\gamma\)$ in \Cref{thm:main}, for which $x_1=0, x_2=\(1+\gamma\)/2$ and $x_3=1$. The stability considerations remain the same for the general case $T>0$ and, as in \Cref{prop:t_zero}, they are fully determined by the sign of $\dot x$ or equivalently, of $f\(x\)$. Summing up, the statements of \Cref{thm:main} and \Cref{prop:t_zero} are illustrated in \Cref{fig:stability}.
\begin{figure}[!htb] 
\begin{subfigure}[b]{0.475\textwidth}
\centering
\resizebox{0.9\textwidth}{!}{%
\begin{tikzpicture}
\begin{axis}[axis x line=none, hide y axis, ymin=-0.001, ymax=0.001] 
\addplot[domain=0:2, samples=201, xtick=\empty]{0}
[every node/.style={yshift=8pt},black]
node[pos=0,xshift=2pt]{$0$}
node[pos=0,xshift=-5pt,yshift=-8pt]{$x$}
node[pos=0,yshift=-5pt,xshift=0.2pt,below] {$\downarrow$}
node[pos=0.02,yshift=-5pt,below] {$\underset{\textstyle x_1}{\phantom{\downarrow}}$}
node[pos=0.5] {$1/2$}
node[pos=0.56,yshift=-2.5pt,below] {$\underset{\textstyle x_2=\(1+\gamma\)/2}{\downarrow}$}
node[pos=1,yshift=-5pt,below] {$\downarrow$}
node[pos=0.98,yshift=-5pt,below] {$\underset{\textstyle x_3}{\phantom{\downarrow}}$}
node[pos=1]{$1$}; 
\addplot[only marks,forget plot, black, mark=|] coordinates {(0,0) (1,0) (2,0)};
\addplot[yshift=-1.5cm,domain=0:2, samples=201]{0}
[every node/.style={yshift=8pt},black]
node[pos=0.28]{$-$} 
node[pos=0.28,yshift=-19pt,below]{$\leftarrow$} 
node[pos=0.78]{$+$}
node[pos=0.78,yshift=-19pt,below]{$\rightarrow$} 
node[pos=0,xshift=2pt]{$0$}
node[pos=0,xshift=-5pt,yshift=-7pt]{$\dot x$}
node[pos=0.07,yshift=-16pt,below] {stable}
node[pos=0.56]{$0$}
node[pos=0.56,yshift=-16pt,below] {unstable}
node[pos=1]{$0$}
node[pos=0.93,yshift=-16pt,below] {stable};
\addplot[only marks,yshift=-1.5cm, forget plot, black, mark=|] coordinates {(0,0) (1.12,0) (2,0)};
\end{axis}
\end{tikzpicture}%
}
\caption{$T=0$.} 
\label{fig:stability_a} 
\vspace{1ex}
\end{subfigure}
\hfill
\begin{subfigure}[b]{0.475\textwidth}
\centering
\resizebox{0.9\textwidth}{!}{%
\begin{tikzpicture}
\begin{axis}[axis x line=none, axis y line=none, ymin=-0.001, ymax=0.001]
\addplot[domain=0:2, samples=201, xtick=\empty]{0}
[every node/.style={yshift=8pt},black]
node[pos=0,xshift=2pt]{$0$}
node[pos=0,xshift=-5pt,yshift=-8pt]{$x$}
node[pos=0.2,yshift=-6pt,below] {$\underset{\textstyle x_1}{\downarrow}$}
node[pos=0.3]{$x_l$} 
node[pos=0.5] {$1/2$}
node[pos=0.56,yshift=-6pt,below] {$\underset{\textstyle x_2}{\downarrow}$}
node[pos=0.75] {$x_u$}
node[pos=0.85,yshift=-6pt,below] {$\underset{\textstyle x_3}{\downarrow}$}
node[pos=1]{$1$}; 
\addplot[only marks,forget plot, black, mark=|] coordinates {(0,0) (0.6,0) (1,0) (1.5,0) (2,0)};
\addplot[yshift=-1.5cm,domain=0:2, samples=201]{0}
[every node/.style={yshift=8pt},black]
node[pos=0,xshift=-5pt,yshift=-7pt]{$\dot x$}
node[pos=0.05]{$+$} 
node[pos=0.05,yshift=-19pt,below]{$\rightarrow$} 
node[pos=0.375]{$-$} 
node[pos=0.375,yshift=-19pt,below]{$\leftarrow$} 
node[pos=0.72]{$+$}
node[pos=0.72,yshift=-19pt,below]{$\rightarrow$} 
node[pos=0.97]{$-$}
node[pos=0.97,yshift=-19pt,below]{$\leftarrow$} 
node[pos=0.2]{$0$} 
node[pos=0.2,yshift=-16pt,below] {stable}
node[pos=0.56]{$0$}
node[pos=0.56,yshift=-16pt,below] {unstable}
node[pos=0.85]{$0$}
node[pos=0.85,yshift=-16pt,below] {stable};
\addplot[only marks,yshift=-1.5cm, forget plot, black, mark=|] coordinates {(0,0) (0.4,0) (1.12,0) (1.7,0) (2,0)};
\end{axis}
\end{tikzpicture}%
}
\caption{$0<T<T_c\(\gamma\)$.} 
\label{fig:stability_b} 
\vspace{1ex}
\end{subfigure}
\vskip\baselineskip
\begin{subfigure}[b]{0.475\textwidth}
\centering
\resizebox{0.9\textwidth}{!}{%
\begin{tikzpicture}
\begin{axis}[axis x line=none, axis y line=none, ymin=-0.001, ymax=0.001]
\addplot[domain=0:2, samples=201, xtick=\empty]{0}
[every node/.style={yshift=8pt},black]
node[pos=0,xshift=2pt]{$0$}
node[pos=0,xshift=-5pt,yshift=-8pt]{$x$}
node[pos=0.2,yshift=-6pt,below] {$\underset{\textstyle x_1}{\downarrow}$}
node[pos=0.3]{$x_l$} 
node[pos=0.5] {$1/2$}
node[pos=0.75,yshift=-6pt,below] {$\underset{\textstyle x_2}{\downarrow}$}
node[pos=0.75] {$x_u$}
node[pos=1]{$1$}; 
\addplot[only marks,forget plot, black, mark=|] coordinates {(0,0) (0.6,0) (1,0) (1.5,0) (2,0)};
\addplot[yshift=-1.5cm,domain=0:2, samples=201]{0}
[every node/.style={yshift=8pt},black]
node[pos=0,xshift=-5pt,yshift=-7pt]{$\dot x$}
node[pos=0.05]{$+$} 
node[pos=0.05,yshift=-19pt,below]{$\rightarrow$} 
node[pos=0.5]{$-$} 
node[pos=0.5,yshift=-19pt,below]{$\leftarrow$} 
node[pos=0.95]{$-$}
node[pos=0.95,yshift=-19pt,below]{$\leftarrow$} 
node[pos=0.2]{$0$} 
node[pos=0.2,yshift=-16pt,below] {stable}
node[pos=0.75]{$0$}
node[pos=0.75,yshift=-16pt,below] {unstable};
\addplot[only marks,yshift=-1.5cm, forget plot, black, mark=|] coordinates {(0,0) (0.4,0) (1.5,0) (2,0)};
\end{axis}
\end{tikzpicture}%
}
\caption{$T=T_c\(\gamma\)$.}
\label{fig:stability_c}
\end{subfigure}
\hfill
\begin{subfigure}[b]{0.475\textwidth}
\centering
\resizebox{0.9\textwidth}{!}{%
\begin{tikzpicture}
\begin{axis}[axis x line=none, axis y line=none, ymin=-0.001, ymax=0.001]
\addplot[domain=0:2, samples=201, xtick=\empty]{0}
[every node/.style={yshift=8pt},black]
node[pos=0,xshift=2pt]{$0$}
node[pos=0,xshift=-5pt,yshift=-9pt]{$x$}
node[pos=0.35,yshift=-6pt,below] {$\underset{\textstyle x_1}{\downarrow}$}
node[pos=0.5] {$1/2$}
node[pos=1]{$1$}; 
\addplot[only marks,forget plot, black, mark=|] coordinates {(0,0) (1,0) (2,0)};
\addplot[yshift=-1.5cm,domain=0:2, samples=201]{0}
[every node/.style={yshift=8pt},black]
node[pos=0.12]{$+$} 
node[pos=0.12,yshift=-19pt,below]{$\rightarrow$} 
node[pos=0.65]{$-$} 
node[pos=0.65,yshift=-19pt,below]{$\leftarrow$} 
node[pos=0.35]{$0$} 
node[pos=0,xshift=-5pt,yshift=-7pt]{$\dot x$}
node[pos=0.35,yshift=-16pt,below] {stable};
\addplot[only marks,yshift=-1.5cm, forget plot, black, mark=|] coordinates {(0,0) (0.7,0) (2,0)};
\end{axis}
\end{tikzpicture}%
}
\caption{$T>T_c\(\gamma\)$}
\label{fig:stability_d}
\end{subfigure}
\caption{Stability properties of the Q-learning dynamics for all values of $T\ge0$. The stability of a steady state (QRE) $x$  (upper axis) is determined by the sign of $\dot x$ (lower axis). The number of QRE and hence, the stability properties change before, at, and after the critical temperature $T_c\(\gamma\)$, cf. \Cref{fig:sub}. The case $T=0$ (upper left panel) is a special case of $0<T<T_c\(\gamma\)$ (upper right panel), for any $\gamma\in\(0,1\)$. For $0<T\le T_c\(\gamma\)$ (panels (b) and (c)), it holds that $x_u,x_2>\(1+\gamma\)/2$. A typical instantiation is given in \Cref{fig:stability_app} in \Cref{app:omitted}.}
\label{fig:stability}\vspace{-0.4cm}
\end{figure}

\begin{remark}[Location of QRE.]\label{rem:location}
From \Cref{thm:main}, it follows that for any $\gamma\in\(0,1\)$, there are no QRE in $[1/2,\(1+\gamma\)/2)$. Also, for any $x\notin [1/2, \(1+\gamma\)/2)$, there exists precisely one $T\ge0$ such that this $x$ is a QRE of \eqref{eq:game}. To see this, we set $f\(x;T,\gamma\)=0$ in \eqref{eq:f_sign} and solve for $T$ as a function for $x$
\[T\(x;\gamma\)=\(2x-1-\gamma\)/\ln{\(\frac{x}{1-x}\)}.\]
For any $x\in[1/2,\(1+\gamma\)/2)$, the expression on the right is negative, which implies that there exists no $T\ge0$ so that this $x$ is a QRE of \eqref{eq:game}. For any other $x\in(0,1)\setminus[1/2,\(1+\gamma\)/2)$, the expression on the right is positive (or equal to zero for $x=\(1+\gamma\)/2$) and yields a unique $T\ge0$ for which this $x$ is a QRE. The points $x=0$ and $x=1$ are QRE (equivalently Nash equilibria) for $T=0$.
\end{remark}


Finally, as can be seen in \Cref{fig:sub2}, the critical level $T_c\(\gamma\)$ is lower for larger values of $\gamma$. Intuitively, the control that needs to be exercised on the system to reach the tipping point at which the number of QRE of the population changes, is lower for larger differences --- as expressed by parameter $\gamma$ --- between the costs of the two technologies. This intuitive property is formally shown in \Cref{cor:monotonicity}.

\begin{corollary}\label{cor:monotonicity}
The critical level $T_c\(\gamma\)$ is decreasing in $\gamma \in \(0,1\)$.
\end{corollary}

\Cref{thm:main} and \Cref{cor:monotonicity} provide a static view --- in terms of the control parameter $T$ --- of the QRE and stability properties of the dynamical system in equation \eqref{eq:dynamics}. From a mechanism design perspective, we are interested in the behavior of the system as $T$ varies according to the control of the central planner. This is discussed next.

\section{Catastrophe Design and Hysteresis Mechanism}\label{sub:hysteresis}

Assuming that the population faces a lock-in at the inefficient state, $x=1$, at which the wasteful and/or costly technology is fully adopted, we can now leverage the stability properties of the behavioral model (Q-learning dynamics) of \Cref{thm:main}, to design a mechanism that will allow the population to overcome this lock-in and adopt the efficient technology, i.e., converge to the $x=0$. The mechanism is described next and visualized in \Cref{fig:main} (cf. \Cref{fig:introduction_plot_2} in \Cref{sec:introduction} and \Cref{app:figure} in \Cref{app:omitted}). Along with the technical details that emerge along the way, the mechanism is formalized in \Cref{cor:convergence}.

\begin{figure}[!hbt]
\centering
\includegraphics[scale=0.58]{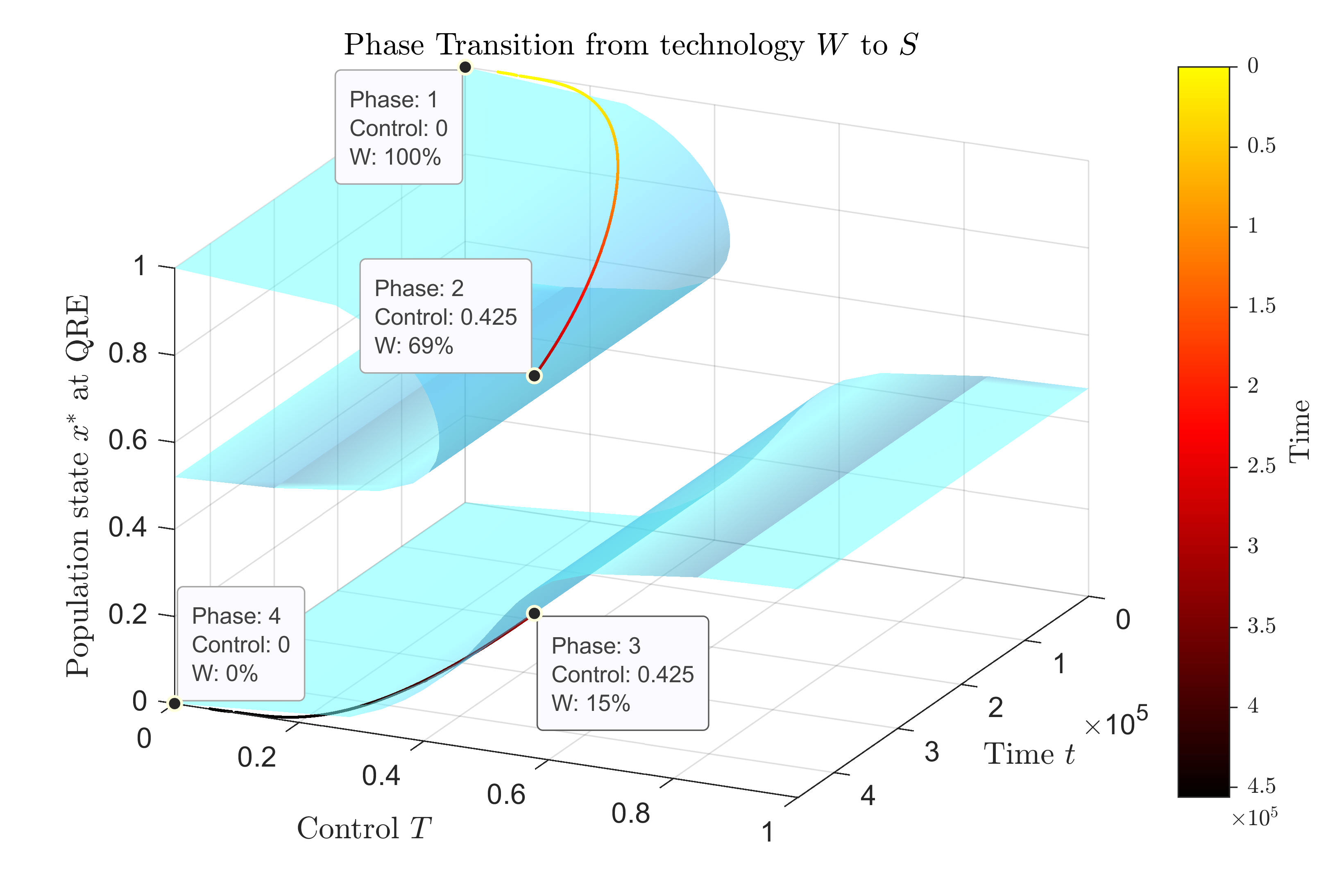}
\caption{Visualization of the combined catastrophe mechanism and hysteresis effect that move the population out of the lock-in on the inefficient technology $W$ to the efficient technology $S$. The QRE surface is projected in time, and the path of the population is depicted by starting from a lighter and gradually moving to darker tones. The phases are described in detail in \Cref{sub:hysteresis}.}
\label{fig:main}\vspace{-0.3cm}
\end{figure}

\begin{description}[leftmargin=0cm]
\item[\emph{Phase 1: Lock-in for $T=0$ and $T<T_c\(\gamma\)$}.] Initially, the population rests at (or close to) the $x=1$ equilibrium at which the whole population adopts the incumbent, wasteful technology. As the control parameter $T$ is gradually increased (but remains below the critical value $T_c\(\gamma\)$), the population stabilizes in the interior QRE on the upper branch of the QRE correspondence (light color line). This QRE is stable, and its attracting region is separated from the stable lower QRE by the unstable middle equilibrium (cf. \Cref{fig:stability_a,fig:stability_b}). \medskip
\item[\emph{Phase 2: Saddle-node bifurcation at $T=T_c\(\gamma\)$.}] As $T$ reaches (from below) the critical va\-lue $T_c\(\gamma\)$, the upper stable QRE and middle unstable QRE collide and amalgamate into a single, unstable equilibrium, cf. \Cref{fig:stability_c}. Intuitively, all this time, the mixed QRE served as a threshold between the attracting regions of the two stable QRE. Hence, although unstable, its existence was critical for the behavior of the system, and its amalgamation with the upper QRE has a critical impact on the dynamics of the system. Together with phase 3, at which the unstable equilibrium is annihilated, this creates a \emph{saddle-node bifurcation, fold bifurcation, or fold catastrophe} in the population dynamics \cite{Str00,Kuz04}.\medskip
\item[\emph{Phase 3: Annihilation of the inefficient QRE for $T>T_c\(\gamma\)$.}] The critical phase transition occurs when the control parameter exceeds (even marginally) the critical level $T_c\(\gamma\)$. The population moves from a QRE with fraction $x>\(1+\gamma\)/2$ of $W$ adopters prior to the saddle-node bifurcation to a QRE with $x<1/2$ immediately after the bifurcation. At this point, the unstable QRE that was generated by the amalgamation of the two upper QRE vanishes, and for any larger value of $T$, the system has only one QRE, which is the single attracting state of the population dynamics (cf. \Cref{fig:stability_d}). The bifurcation that annihilates the amalgamated QRE and the subsequent stabilization of the dynamics in the single remaining QRE is the critical change of the system. In response to a small change in the control parameter (from below to slightly above the critical level), the population moves (abruptly) out of a state (QRE) in the attracting region of the inefficient equilibrium $x=1$ to a QRE in the attracting region of the efficient equilibrium $x=0$ for $T=0$.
\medskip
\item[\emph{Phase 4: Hysteresis effect as $T\to0$.}] The (permanent) impact of the bifurcation on the population dynamics is that (the bifurcation) encodes a new \emph{memory} to the population. In particular, as explained in the previous phase, for any value of $T>T_c\(\gamma\)$, the system stabilizes at a QRE with fraction $x<1/2$ of $W$ adopters. Consequently, using this QRE as a starting point and reducing (gradually or even instantaneously) the control parameter $T$ back to $0$, the system provably converges to the efficient equilibrium, i.e., to the adoption of the efficient technology. This creates a \emph{hysteresis effect}: when the control parameter $T$ is back to its initial level, i.e., $T=0$, at which no exogenous control is exercised to the system, the population equilibrates at a different state than its initial one.
\end{description}

The above mechanism, termed \emph{catastrophe mechanism}, is formalized in \Cref{cor:convergence}, whose proof is given in \Cref{app:omitted}.  \Cref{cor:convergence} is stated for initial conditions $x_0\in[\(1+\gamma\)/2,1]$, i.e., for initial conditions that capture the lock-in at the inefficient technology. For $x_0\in[0,\(1+\gamma\)/2)$ convergence to the $x=0$ equilibrium is trivial.

\begin{theorem}[Catastrophe mechanism.]
\label{cor:convergence}
For a fixed cost parameter $\gamma \in \(0,1\)$ consider the population game \eqref{eq:game} and the revision protocol \eqref{eq:protocol}, which together lead to the population evolutionary dynamics $\dot x$ in \eqref{eq:dynamics}. Then, for any initial population state $x_0\in[\(1+\gamma\)/2,1]$, there exists a finite sequence $\langle T_0,T_1,\dots, T_n\rangle$ of control parameter values with $T_0=T_n=0$, and $\max_{i\le n}{T_i}>T_c\(\gamma\)$, so that the iterative procedure which, starting from phase $i=0$,
\begin{itemize}
\item scales the control parameter $T$ to $T_i$,
\item allows the system to converge to a QRE, $x^*\(T_i\)$, 
\end{itemize}
and reiterates for phase $i+1$, generates a sequence of population states (QRE) that converges to the risk-dominant equilibrium, $x^*\(T_n\)=0$, at which the population adopts the efficient technology.
%
\end{theorem}

\begin{remark}\label{rem:sequence}
 In theory, it suffices to use the sequence $\langle T_0=0, T_1>T_c\(\gamma\),T_n=0\rangle$. However, for practical applications, such abrupt changes of the control parameter $T$ may not be possible. For instance, if $T$ denotes a tax or temperature parameter, then only gradual or smooth changes in $T$ may be permissible. In such cases, the sequence $\langle T_0, T_1,\dots, T_n\rangle$ may require additional steps, but as long as it satisfies the requirements of \Cref{cor:convergence}, the final outcome of the catastrophe mechanism still obtains. 
\end{remark}

\paragraph{Properties and Practical Implementation}
From the above, it is immediate that the proposed mechanism has two desirable properties for practical applications. First, the critical population mass, $x^*\(T_c\(\gamma\)\)$, i.e., the fraction of adopters of the wasteful technology at the critical level $T_c\(\gamma\)$, is strictly larger than the mixed equilibrium $\(1+\gamma\)/2$ of the initial game. This means that, under the proposed mechanism, the bifurcation occurs when the wasteful technology is still adopted by the \emph{majority} of the population. The deviation to the new technology by a minority of the population is sufficient to trigger the adoption of the new technology by the whole population and become permanent. Indicatively, in the depicted scenario of \Cref{fig:main}, the population transitions from $69\%$ of $W$ adopters immediately before the saddle-node bifurcation (or catastrophe) to $15\%$ of $W$ adopters immediately after. \par
Second, and more importantly, from a social planner's perspective, the required influence on the control parameter needs to be only temporary. Yet, due to the resulting hysteresis, the effects of the mechanism are permanent. Assuming that the exercise of any influence on the system incurs a cost, this provides a desirable trade-off between the resources (expenses/time) of implementing such a mechanism and the duration of its impact. \par
Finally, related to the previous property is also the interpretation of the initial normalization of $VK$ to $1$, cf. equation \eqref{eq:dynamics}, in monetary terms. Since this normalization is equivalent to dividing equation \eqref{eq:dynamics_pre} with $VK$, all instances of $\gamma$ and $T$ in the previous analysis should be interpreted as $\gamma/VK$ and $T/VK$, respectively, for practical purposes. For instance, the threshold $T_c\(\gamma\)<1/2$ for any $\gamma\(0,1\)$ in \Cref{thm:main} implies that in applications with $VK\neq1$ (typically $VK\gg1$), the tipping point satisfies $T_c<VK/2$. An illustrative application of this mechanism at which $T$ is interpreted as taxation is presented in \Cref{sec:applications}.

%% file: conclusions.tex
\section{Conclusions}\label{sec:conclusions}
In this paper, we developed and studied a mechanism to destabilize lock-in at inefficient equilibria in (population) coordination games. While requiring only temporary influence to the population with limited implementation costs, the proposed mechanism achieves a permanent transition from a sub-optimal locked-in equilibrium state, e.g., the adoption of a status quo, yet wasteful technology, to a socially and individually preferable state, e.g., the adoption of an innovative and sustainable technology. As a first step, our methods expose the geometric features of the Q-learning dynamics in the widely studied model of population games. The transition from the inefficient to the efficient equilibrium is then achieved via an induced saddle-node bifurcation or controlled catastrophe in the learning dynamics. The resulting mechanism is both simple to understand and easy to implement, yet its provable guarantees make critical use of the complex underlying geometry.  
\par
The lock-in at inefficient equilibria is a long-standing problem with ample manifestations across the whole spectrum of economic, social, natural, and behavioral sciences. However, little is known on how to overcome the lock-in in practical applications. Most existing studies report experimental findings \cite{Bra06,Bra16,Mas17} (and references therein) with only a few recent exceptions \cite{Yan17,Mas19}. As economies and social interactions become increasingly complex and unpredictable, the design of effective mechanisms, even in purely adversarial conditions (limited resources and ability to control the population), is more relevant than ever. Aiming to make progress in this direction, our work provides a (first to our knowledge) provably efficient mechanism to overcome this lock-in and opens several directions for future work. Some of the most naturally arising questions concern the implementation of the current mechanism --- or modifications thereof --- in adversarial settings at which the population can control some parameter in response to changes in the control parameter $T$ (here that would be parameter $\gamma$) or under modified behavioral models. Testing the efficiency of the proposed mechanism in field or laboratory experiments is another promising direction for study.

%% file: robustness.tex
\section{Structural Robustness of the Model}\label{sec:robustness}
Thus far, our analysis and the proposed mechanism to move out of the inefficient lock-in rested on two (seemingly) restrictive assumptions. The first is that the updates in the population state are not subject to uncontrolled perturbations and the second is that $\alpha=2$, i.e., that the parameter which expresses the intensity of network effects is such that the system can be represented as an evolutionary game, cf. \Cref{sub:elementary}. For practical purposes and more generally for systems with constantly changing or unknown characteristics, it is important to show that the efficiency of the suggested mechanism is not tied to these particular assumptions. It turns out that this indeed turn and our goal in this section is to test and prove this claim in two directions. \par
First, we show that the proposed mechanism continues to apply under state dependent perturbations in the Q-learning dynamics without significant impact in its actual implementation, cf. \Cref{sub:robust_e}. This holds for any reasonably -- in the sense that the model parameters remain in their admissible regions -- bounded perturbations.\par
Second, we study the robustness of the model for network effects of different intensity, as expressed by different values of parameter $\alpha$, cf. \Cref{eq:value}. In this direction, we obtain a conservative (yet tight for extreme values of the parameters) theoretical bound on the critical value $T_c\(\gamma\)$, i.e., on the control that needs to be exercised to the system to trigger the desired transition out of the attracting region of the initial technology. In practical terms, as $\alpha$ increases, a network split becomes more damaging and the population reaches the tipping point for increasingly lower values of $T_c$. In this respect, the case $\alpha=2$ that we treated in the previous part is one of the costliest to destabilize. Finally, even though the number of QRE in the attracting region of the efficient technology may change for large values of $\alpha$, the mechanics of the proposed mechanism essentially remain unaffected. 

\subsection{State Dependent Perturbations}\label{sub:robust_e}
In general, small perturbations in a dynamical system can have major effects in its outcome, cf. \cite{Str00,Puu91} or \cite{Pal17} among many others. To study the behavior of the currently used Q-learning dynamics in equation \eqref{eq:dynamics} in response to such perturbations, we consider a state-dependent noise term $\epsilon\(x\)$
\begin{equation}\label{eq:perturbed}
\dot x =x\(1-x\)\lt 2x-1-\gamma-T\ln{\(\frac{x}{1-x}\)}+\epsilon\(x\)\rt, \qquad \text{for }x \in\(0,1\).
\end{equation}
Similarly to equation \eqref{eq:f_sign}, it will be convenient to introduce the following notation.

\begin{notation}\label{not:3}
For fixed $\gamma\in\(0,1\)$ and $T\ge0$, let 
\begin{equation}\label{eq:f_sign_e}
f_\epsilon\(x;T,\gamma\):= 2x-1-\gamma-T\ln{\(\frac{x}{1-x}\)}+\epsilon\(x\), \qquad \text{for }x\in\(0,1\).
\end{equation}
where $\epsilon\(x\):\(0,1\)\to \mathbb R$ is a noise-function that depends on the current population state and which for all $x\in\(0,1\)$ satisfies $|\epsilon\(x\)|\le \epsilon_0$ for some constant $\epsilon_0 \in \(0,\min{\{\gamma,1-\gamma\}}\)$. In particular, $f_\epsilon\(x;T,\gamma\):= f\(x;T,\gamma\)+\epsilon\(x\)$, where $f\(x;T,\gamma\)$ is defined in equation \eqref{eq:f_sign}. Whenever obvious from the context, we will omit the dependence on $T,\gamma$ and write $f_\epsilon\(x\)$. 
\end{notation}
The bound $\epsilon_0$ is naturally imposed by the requirement that the game parameter, $\gamma$, should remain within its admissible region. In particular, the bound $\epsilon_0\in \(0,\min{\{\gamma,1-\gamma\}}\) $ on $|\epsilon\(x\)|$ ensures that the perturbation (noise) is not large enough to change the nature of the system in favor of the originally undesirable equilibrium.\par
Turning to the dynamics of the $\epsilon\(x\)$-perturbed system, the number of interior steady states of the dynamics, as expressed by the roots of $f_\epsilon\(x;T,\gamma\)$, may change significantly in comparison to the unperturbed system. To see this, assume that $x^*$ is a steady state of the unperturbed system for some $T>0$, so that $f\(x^*;T,\gamma\)=0$. Then
\[f_\epsilon\(x^*;T,\gamma\)=\epsilon\(x^*\)\]
and hence, there exists a local neighborhood around $x^*$ for which the behavior of the system is essentially determined by the noise term. Accordingly, $\dot x$ may have an arbitrary number of sign-changes in this neighborhood depending on the exact shape of $\epsilon\(x\)$ and hence, it is not possible to argue about the exact state of the system within this region. However, we can still reason about the stability of the system \emph{outside} these (bounded) neighborhoods in the same fashion as we did for the unperturbed system. In fact, the stability analysis, cf. \Cref{thm:main} and \Cref{fig:stability}, carry over with the only difference that the system will now stabilize within a \emph{neighborhood} of the initial QRE, i.e., of the QRE of the unperturbed system. This allows us to establish equivalent convergence guarantees and retain the desirable effect of the proposed catastrophe mechanism. This is formalized in \Cref{thm:stability_e} whose proof can be found in \Cref{app:section6}.

\begin{theorem}[Catastrophe mechanism under bounded perturbations.]\label{thm:stability_e}
For a fixed cost parameter $\gamma \in \(0,1\)$ consider the population game \eqref{eq:game}, the revision protocol \eqref{eq:protocol} and the perturbed resulting dynamics in equation \eqref{eq:perturbed}, where $|\epsilon\(x\)|\le \epsilon_0\in\(0,\min{\{\gamma,1-\gamma\}}\)$ is a state-dependent noise term for all $x\in\(0,1\)$. Then, the catastrophe mechanism of \Cref{cor:convergence} still applies. \par
In particular, for any initial population state $x_0\in[\(1+\gamma\)/2,1]$, there exists a finite sequence $\langle T_0,T_1,\dots, T_n\rangle$ of control parameter values with $T_0=T_n=0$, and $\max_{i\le n}{T_i}>T_c\(\gamma-\epsilon_0\)$, so that the iterative procedure which, starting from phase $i=0$,
\begin{itemize}
\item scales the control parameter $T$ to $T_i$,
\item allows the system to converge to a neighborhood of the QRE, $x^*\(T_i\)$, of the unperturbed system,
\end{itemize}
and reiterates for phase $i+1$, generates a sequence of population states (QRE) that converges to the risk-dominant equilibrium, $x^*\(T_n\)=0$, at which the population adopts the efficient technology. 
\end{theorem}

\Cref{thm:stability_e} suggests that the results of the unperturbed case, largely carry over also to the perturbed case. Intuitively, while the exact number of steady states in the perturbed system cannot be determined, the new steady states are all located in some bounded neighborhood of the old steady states. \footnote{The exact size of these neighborhoods depends on the value of $\epsilon\(x\)$ at $x\in\(0,1\)$. However, it can be shown that these neighborhoods shrink for values of $x$ close to the boundary or for values of $x$ close to $1/2$, cf. Proof of \Cref{thm:stability_e} in \Cref{app:omitted} and \Cref{fig:perturbation} for an illustration.} Outside these regions, the sign of $\dot x$ remains the same as in the unperturbed case and hence, the dynamics converge to these regions in the same fashion that they converged to the exact steady states in the unperturbed case. Two illustrations --- of bounded and unbounded state dependent perturbations --- are given in \Cref{fig:perturbation}.

\begin{figure}[!htb]
\centering
\begin{minipage}[t]{0.49\textwidth}
\centering
\includegraphics[width=\textwidth]{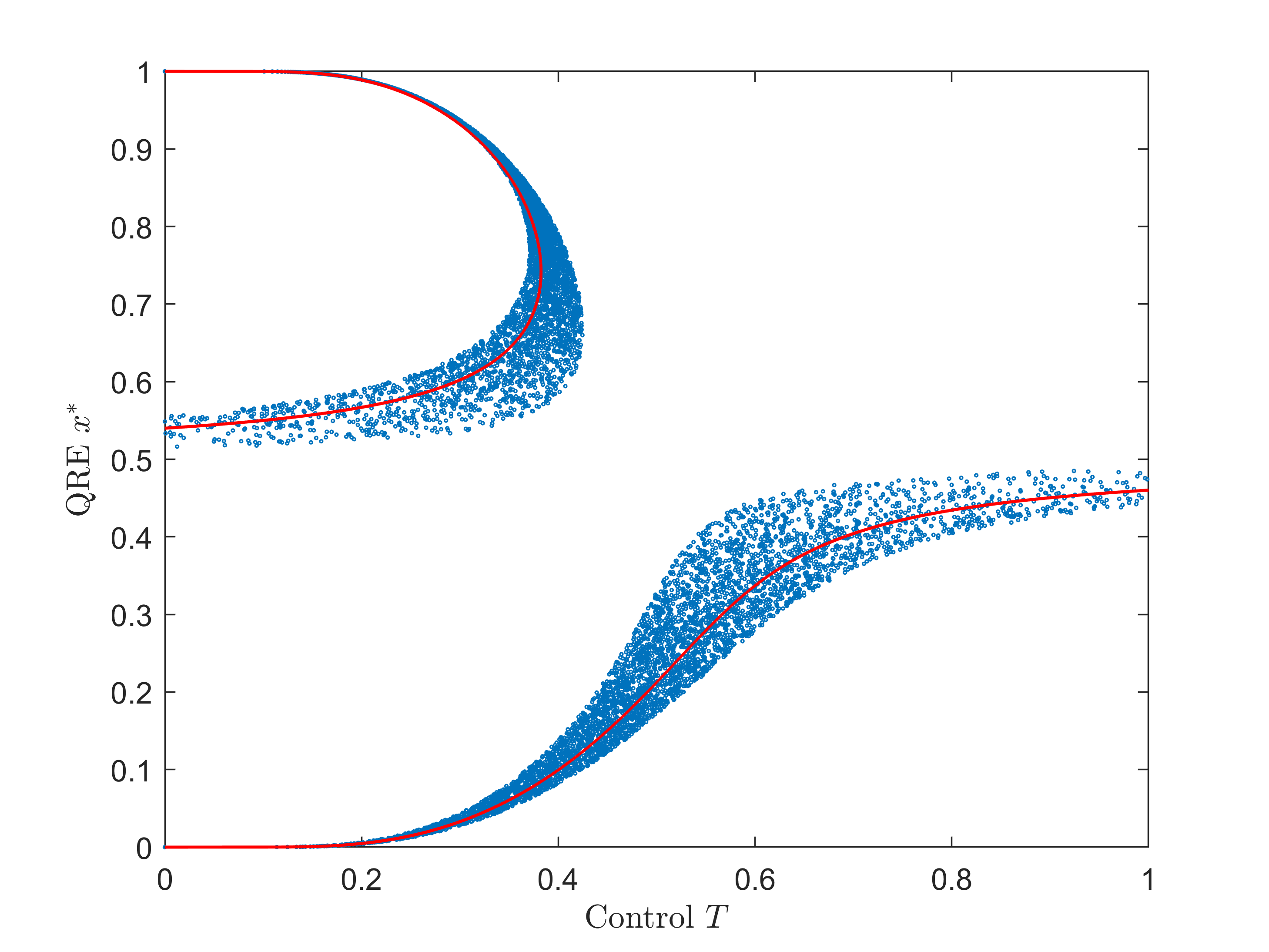}
\end{minipage}\hfill
\begin{minipage}[t]{0.49\textwidth}
\centering
\includegraphics[width=\textwidth]{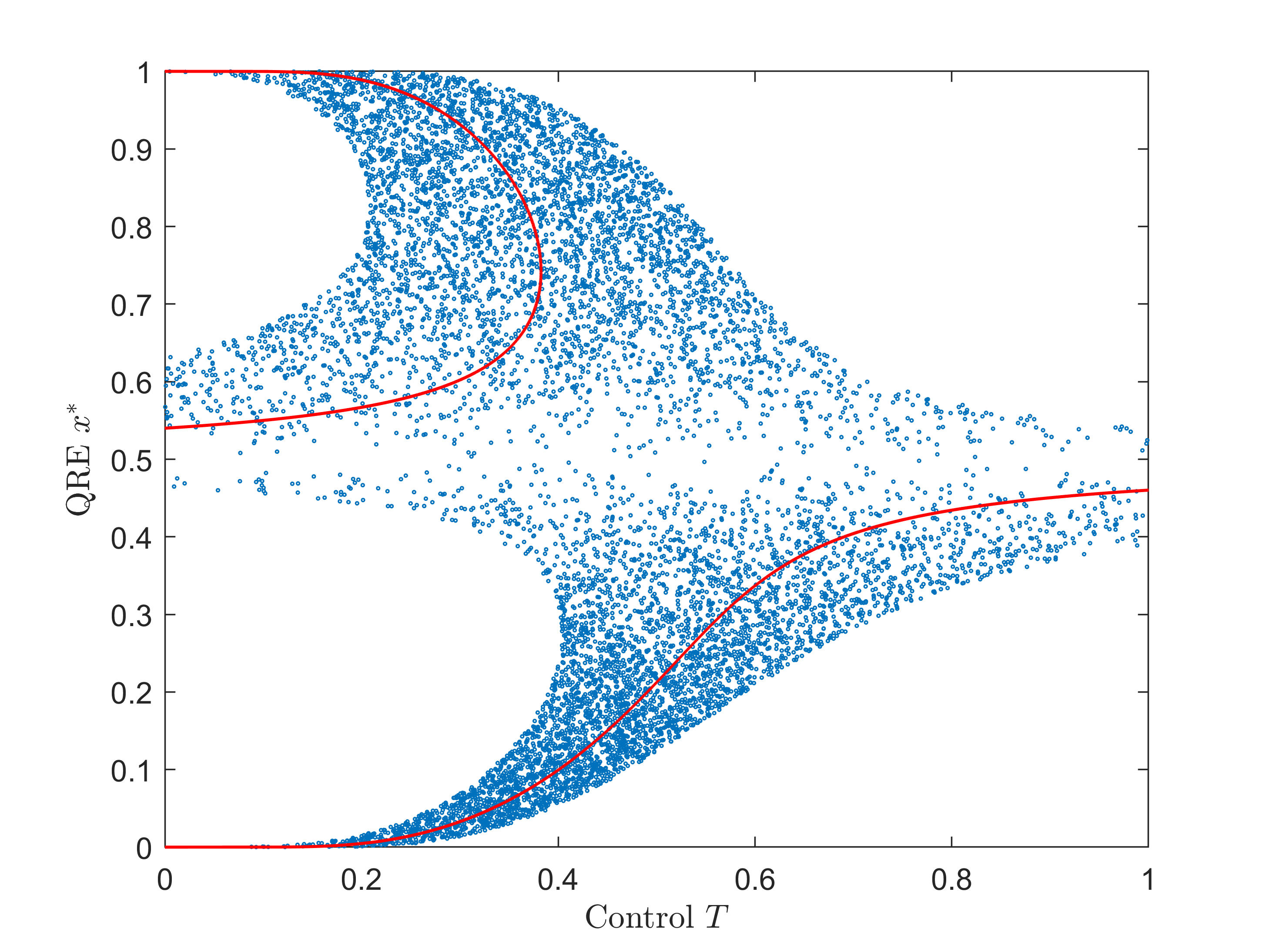}
\end{minipage}
\caption{QRE (red lines) of the unperturbed system for $\gamma=0.08$ and QRE (blue dots) under state-dependent perturbations $\epsilon\(x\)=0.01+\text{rand}\cdot x^{1/6}\(1-x\)$ (left panel) and $\epsilon\(x\)=-\text{rand}\cdot\ln(x-x^{0.5})$ (right panel), where $\text{rand}$ denotes a random number in $[0,0.1]$. In both panels, the largest deviations occur for population states close to $1/2$, cf. equation \eqref{eq:bound}, whereas for values of $x$ close to the boundaries, the perturbed systems reduce to the original system (red line) as the perturbations get dominated by the entropy term. In the left panel, which satisfies the assumptions of \Cref{thm:stability_e}, the proposed mechanism can be implemented with the only difference that, now, the dynamics can converge to any blue dot (instead of the red line) that lies in a neighborhood of the original QRE (in the unperturbed system). By contrast, the perturbation in the right panel does not satisfy the assumptions and the outcome of the mechanism is ambiguous as the population may not transition to states below $1/2$ even for large values of the control parameter $T$.}
\label{fig:perturbation}
\end{figure}

\begin{remark}\label{rem:mechanism}
From the perspective of mechanism design, it is also meaningful to study the effects of the perturbation on the implementation of the proposed mechanism and in particular, on the control that needs to be exercised on $T$. To quantify this change, we solve equation $f\(x\)=0$ for $T$, cf. \Cref{rem:location}, and recover the control parameter as a function of $x$ on the geometric locus of all QRE of the unperturbed system
\[T\(x;\gamma\)=\(2x-1-\gamma\)/\ln{\(\frac{x}{1-x}\)}.\]
Similarly, if $T_\epsilon\(x;\gamma\)$ denotes the value of the control parameter on the geometric locus of all QRE of the perturbed system, then
\begin{align}\label{eq:t_pert}
T_{\epsilon}\(x;\gamma\)&=\frac{2x-1-\gamma+\epsilon\(x\)}{\ln{\(\frac{x}{1-x}\)}}=T\(x;\gamma\)+\frac{\epsilon\(x\)}{\ln{\(\frac{x}{1-x}\)}}
\end{align}
which implies that if some $x\in\(0,1\)$ is a QRE of both systems, then
\begin{equation}\label{eq:bound}
|T_{\epsilon}\(x;\gamma\)-T\(x;\gamma\)|\le \epsilon_0\cdot\left|\ln{\(\frac{1}{x}-1\)}\right|^{-1}.
\end{equation}
Equation \eqref{eq:bound} highlights an important property that comes from the inclusion of the entropy term in the dynamics. Namely, as $x$ approaches the boundary for decreasing $T>0$, the effect of any bounded perturbation on the control variable gets increasingly dominated by the entropy term. At $T=0$, the system stabilizes again at $x=0$. An illustration is given in \Cref{fig:perturbation}. Finally, for an alternative bound on the maximum additional control that needs to be exercised on $T$ to trigger the critical phase transition, equation \eqref{eq:t_pert} and the monotonicity of the critical level $T_c\(\gamma\)$ in $\gamma$, cf. \Cref{cor:monotonicity}, imply that $T_c\(\gamma;\epsilon\)\le T_c\(\gamma-\epsilon_0\)$, where $T_c\(\gamma;\epsilon\)$ denotes the critical value in the perturbed system (cf. \Cref{thm:stability_e}).
\end{remark}

\subsection{Intensity of Network Effects}\label{sub:robust_a}
Parameter $\alpha$ in \Cref{eq:value} captures the intensity of the network effects, i.e., the value that can be created by each technology as a function of its adoption by the population of investors. Values of $\alpha<1$ indicate subadditive value, i.e., that the population payoff is maximized when the network splits between the two technologies, and fall out of the present scope. Similarly, when $\alpha=1$, the degree of adoption does not affect the generated value and the efficient technology constitutes a dominant strategy which trivializes the resolution of the game. \par
In the present context, we are interested in cases with superadditive value --- or direct positive network effects --- that are expressed by values of $\alpha>1$ (an illustration is provided in \Cref{rem:alpha}). Thus far, we have assumed that $\alpha=2$ mainly for expositional purposes, however, in this part, we show that the results generalize essentially unaltered to any $\alpha>1$. \par
Specifically, let $\alpha>1$. Then by equations \eqref{eq:payoffs} and \eqref{eq:dynamics_pre}, we otain the relationship
\[\dot x =x\(1-x\)\lt VK^{\alpha-1}x^{\alpha-1}-\gamma-VK^{\alpha-1}\(1-x\)^{\alpha-1}-T\ln{\(\frac{x}{1-x}\)}\rt,\]
which after normalizing $VK^{\alpha-1}$ to $1$ --- or equivalently after dividing the equation with $VK^{\alpha-1}$ and setting $\gamma\to \gamma/VK^{\alpha-1}$ and $T\to T/VK^{\alpha-1}$ --- yields the dynamics
\begin{equation}\label{eq:all}
\dot x =x\(1-x\)\lt x^{\alpha-1}-\gamma-\(1-x\)^{\alpha-1}-T\ln{\(\frac{x}{1-x}\)}\rt.
\end{equation}
The geometric locus of the steady states (QRE) of the dynamics in equation \eqref{eq:all} is illustrated for various values of $\alpha>1$ and $T\ge0$ in \Cref{fig:alpha}. As network effects become more intense, i.e., as $\alpha$ increases, splits of the population become increasingly detrimental for both aggregate and individual welfare. As a consequence, the system becomes more sensitive to the control $T$ and the inefficient equilibrium is easier to destabilize. This can be seen from the thinning upper component in the QRE surface in the left panel of \Cref{fig:alpha}. Moreover, for higher values of $\alpha$, the lower component of the QRE correspondence develops an unstable part in the $[0,1/2]$ region (dashed part in the red and green lines in the left panel and horn-shaped dark area in the right panel) at which the population may oscillate between three or marginally two QRE. However, all these QRE lie in the attracting region of $x=0$ for $T=0$ and hence, this instability does not compromise the outcome of the proposed catastrophe mechanism. In particular, whenever $T\ge1/2$, there exists a unique QRE $x^*\in \(0,1/2\)$ (light area in the right panel). This implies that by increasing $T$ to (at most) $1/2$, the system can be stabilized in a state $x^*$ which lies in the attracting region of the desired equilibrium $x=0$. Subsequently, the last phase --- at which the control parameter is reset back to $0$ --- of the proposed catastrophe mechanism in \Cref{cor:convergence} can be applied unaltered (e.g., by reducing $T$ back to $0$) and still yield the desired outcome of convergence to the $x=0$ equilibrium, independently of the intensity of the network effects, i.e., of the value of $\alpha>1$. This is formalized in \Cref{thm:unique}.

\begin{figure}[!htb]
\centering
\begin{minipage}[t]{0.51\textwidth}
\centering
\includegraphics[width=\textwidth]{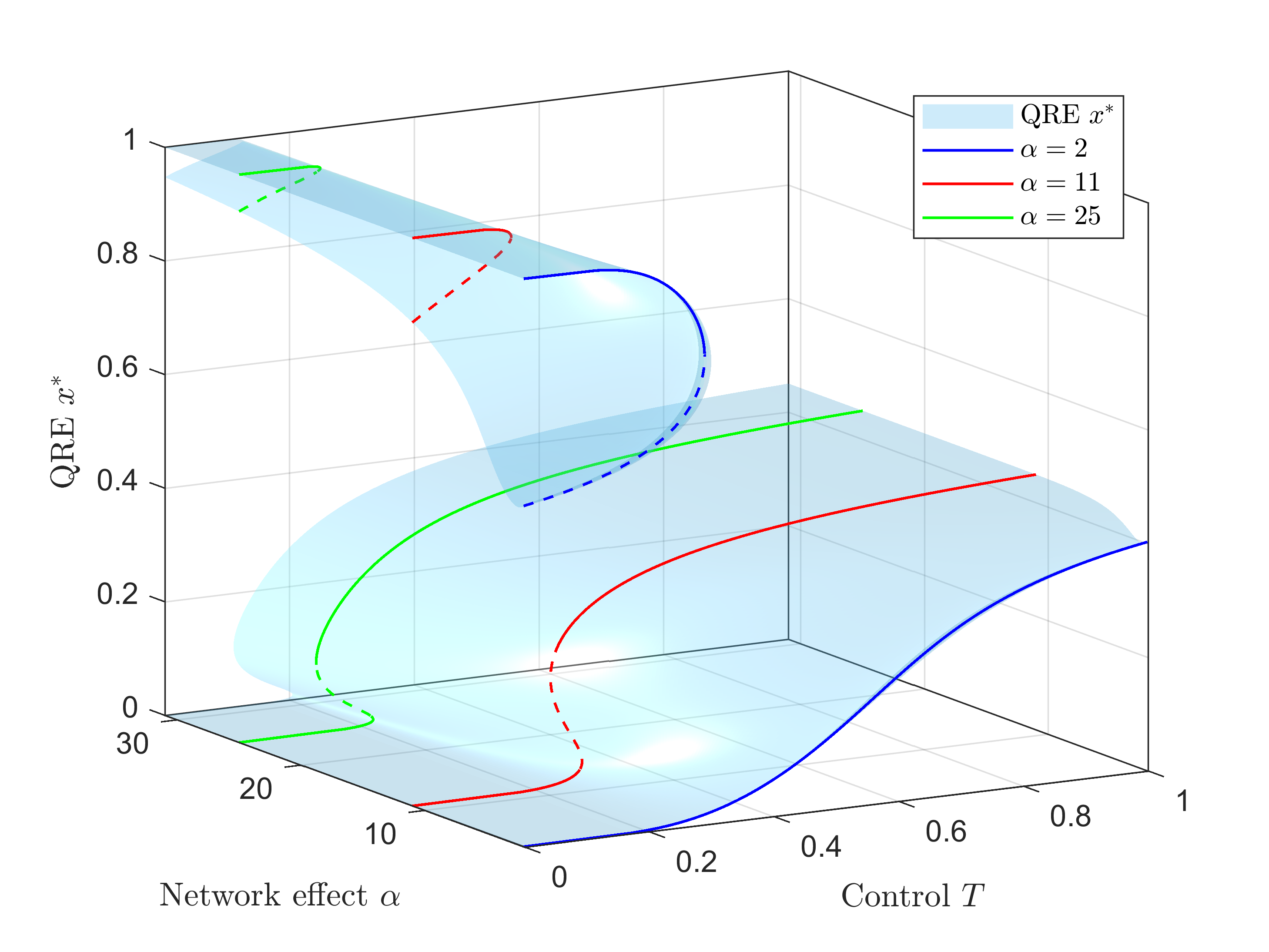}
\end{minipage}\hfill
\begin{minipage}[t]{0.49\textwidth}
\centering
\includegraphics[width=\textwidth]{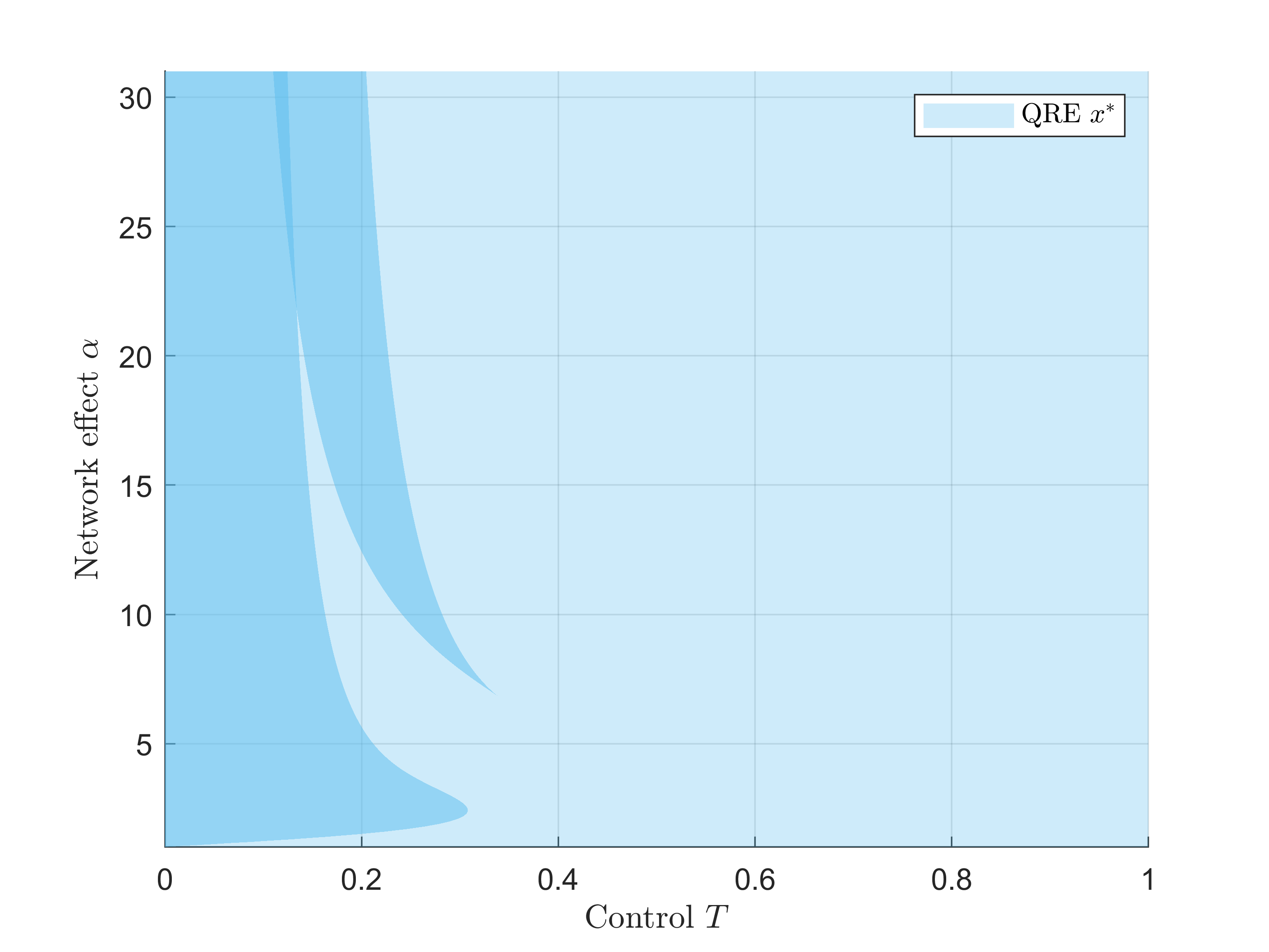}
\end{minipage}
\caption{Left panel: the QRE correspondence for various values of $\alpha$ ($\gamma=0.2$ is fixed). The case $\alpha=2$ (blue slice) has been treated in \Cref{sec:results}. As $\alpha$ increases, the lower component of the QRE correspondence develops an unstable region which can be seen as an s-shaped dashed part in the green slice (left panel). Right panel: the projection of the QRE surface on the $\(\alpha,T\)$ plane. Darker areas correspond to parameter values with more QRE. In particular, the bifurcations between one and three QRE occur at the boundary between dark and light areas at which the system has two QRE. The parameter range with three QRE that develops for larger values of $\alpha$ can be seen as a dark horn-shaped area in the upper/upper left part of the graph. For certain parameter values, there exist more QRE both in the upper and in the lower component of the QRE surface (darker area at which the two initially dark areas overlap). This unstable area is inconsequential in the transition to the efficient technology since it lies entirely in the attracting region of the $x=0$ equilibrium, cf. \Cref{thm:unique}. Finally, it can be seen that $T_c\(\gamma\)$ is decreasing in $\alpha$ for values of alpha larger than approximately 2.5 and that the bound $T\ge1/2$ is (in all cases) particularly conservative.}
\label{fig:alpha}
\end{figure}

\begin{theorem}[Catastrophe mechanism for arbitrary network effects]\label{thm:unique}
For a fixed cost parameter $\gamma \in \(0,1\)$ consider the population game defined by the payoff functions in equation \eqref{eq:payoffs} and the revision protocol \eqref{eq:protocol} which together lead to the population evolutionary dynamics $\dot x$ in equation \eqref{eq:all}. Then, for any $\alpha>1$, the catastrophe mechanism of \Cref{cor:convergence} still applies. \par
In particular, for any initial population state $x_0\in[\(1+\gamma\)/2,1]$, there exists a finite sequence $\langle T_0,T_1,\dots, T_n\rangle$ of control parameter values with $T_0=T_n=0$, and $\max_{i\le n}{T_i}\ge1/2(>T_c\(\gamma\))$, so that the iterative procedure which, starting from phase $i=0$,
\begin{itemize}
\item scales the control parameter $T$ to $T_i$,
\item allows the system to converge to a QRE, $x^*\(T_i\)$,
\end{itemize}
and reiterates for phase $i+1$, generates a sequence of population states (QRE) that converges to the risk-dominant equilibrium, $x^*\(T_n\)=0$, at which the population adopts the efficient technology. 
\end{theorem}


\begin{remark}
As mentioned above, the only difference in comparison to \Cref{cor:convergence} is that there exists a critical value $\alpha_c$, so that for values $\alpha>\alpha_c$, there exist multiple QRE in the interval $x\in\(0,1/2\)$ for certain values of $T<1/2$, cf. \Cref{fig:alpha}.\footnote{Technically, this is another bifurcation in the number of QRE, this time caused by parameter $\alpha$. However, we do not further explore this bifurcation since this deviates from the scope of the current section.} In practice, the bound $T=1/2$ --- to obtain a unique QRE, $x^*\in(0,1/2)$ --- is extremely conservative, and essentially, it is only tight for the absolutely extreme cases of $\gamma$ close to $0$ and $\alpha=1$ (not depicted here). Moreover, since convergence to any QRE $x^*\in \(0,1/2\)$ suffices to move the system out of the inefficient lock-in, \Cref{thm:unique} can be equivalently stated for $\max_{i\le n}T_i\ge T_c\(\gamma\)$ without compromising its outcome. \par
From the perspective of policy-making, it is also worth noting that the case $\alpha=2$ that was used up to now corresponds to one of the costliest cases to treat in practice. As can be seen in the right panel of \Cref{fig:alpha}, the critical level $T_c$ at which the bifurcation in the upper part of the QRE surface occurs --- i.e., the point at which the two upper QRE merge and immediately thereafter, disappear --- is decreasing in $\alpha$ for values of $\alpha$ larger than approximately $2.5$ (protrusion of the dark blue area in the bottom left corner). In particular, larger values of $\alpha$ capture technologies with more intense (positive) network effects. The more intense effects translate to smaller values of the tipping point $T_c\(\gamma\)$, which, in turn, implies a lower cost to implement the proposed mechanism. 
\end{remark}

%% file: applications.tex
\section{Application: Proof of Work Blockchain Mining and Taxation}\label{sec:applications}

Since the launch of Bitcoin (BTC) by the pseudonymous Satoshi Nakamoto, \cite{Na08}, Proof of Work (PoW) blockchains and their applications --- most notably cryptocurrencies --- have taken the world by storm. Widely considered as a revolutionary technology, blockchains have attracted the attention of institutions, technology-corporations, investors and academics. In these networks, self-interested miners ensure the proper functionality of the supported applications by exerting effort and receiving monetary incentives in return \cite{Leo20}. To participate, and without any centralized authority to grant permission, miners are required to provide evidence of some predefined, and typically scarce, resource, e.g., computational power that utilizes electricity in Proof of Work (PoW) protocols, \cite{Fia19,Gor19}, or units of the native cryptocurrency in Proof of Stake (PoS) protocols \cite{Bro19}. \par
\begin{figure}[!hbt]
\centering
\includegraphics{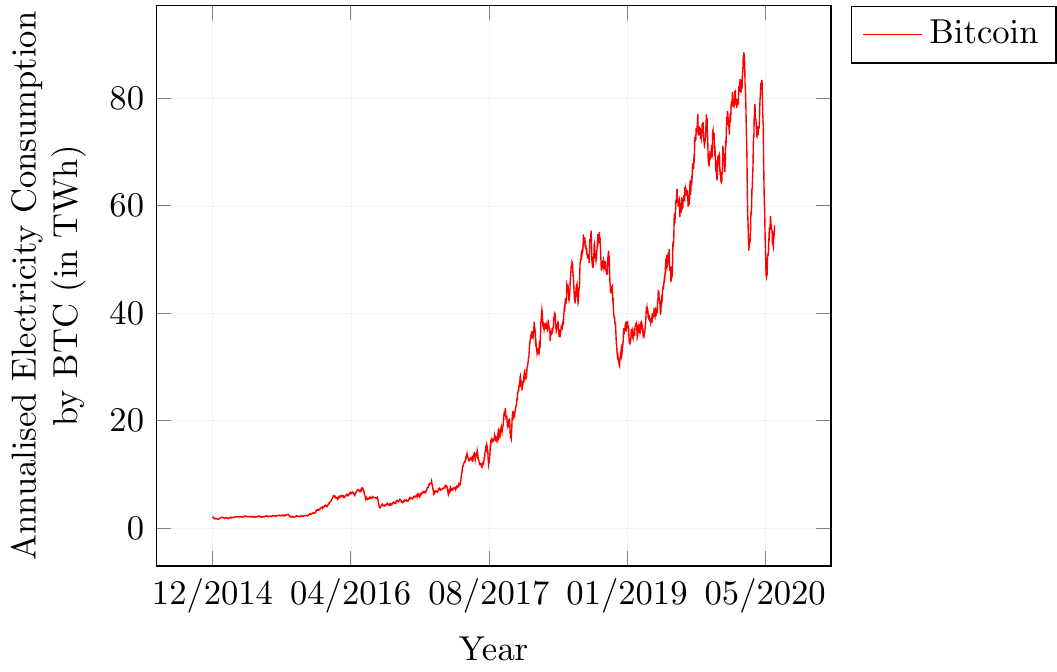}
\caption{The Cambridge Bitcoin Electricity Consumption Index (CEBCI) showing the estimate of the mean energy consumption of the Bitcoin (BTC) network in TWh (annualized). From less than 20 TWh in August of 2017, the total network consumption exceeded 90 TWh during the first days into 2020.\protect\footnotemark Subsequent fluctuations can be To understand the extent to which these figures potentially underestimate future energy consumption, they still reflect a time point in which BTC has not been widely endorsed by governments and public institutions and they do not account for the numerous PoW cryptocurrencies beyond BTC. Source: \href{https://www.cbeci.org/}{CBECI by University of Cambridge, Judge Business School}.}
\label{fig:waste}
\end{figure}

\footnotetext{Interestingly, the subsequent (extreme) fluctuations of the CEBCI index within 2020 co-occur with the global COVID-19 disease outbreak.}

Yet, PoW blockchains networks grow orthogonal to their design philosophy in terms of their environmental footprint. Currently, one BTC transaction wastes as much energy --- in terms of carbon footprint --- as 775.818 VISA transactions \cite{Dig20}. More alarming than its current levels --- which rank the BTC network above Finland and Pakistan in terms of electricity consumption --- is the consumption's increasing trend: the electricity used by the BTC network approximately doubles every year \cite{Cbe20}. The total figures only get worse, if we account for all other PoW blockchains such as Ethereum \cite{But19}.\par

These alarming figures call for mechanisms to accelerate the development and, more crucially, the adoption of alternative technologies such as (PoS) \cite{Bro19,Haz19}, also known as \emph{virtual mining} \cite{Ben16}.\footnote{Both from theoretical and practical perspective, PoS technologies have shown to offer strong guarantees analogous to those of PoW \cite{Kia17,Gar15,But19}.} Yet, an important barrier is that the value of a cryptocurrency --- or in general the reliability of its applications --- depends on the size of its mining network (with larger network implying higher safety). More mining power implies that it is more costly for a potential attacker to gather the required resources and compromise the functionality of the blockchain, \cite{Kia16,Bro19}. Hence, when the rest of the population mines a specific PoW cryptocurrency, then it is individually rational (preferable) for any single miner to also mine that cryptocurrency. \par
To foster the transition to the existing sustainable alternatives, like PoS, using the currently developed framework, the control parameter $T\ge0$ can be interpreted as \emph{taxation}\footnote{Earlier interpretations of $T$ as taxation parameter can be found in \cite{Wol12} and \cite{Yan17} among others.}. As mentioned in \Cref{sec:introduction}, changes in $T$ are essentially equivalent to rescaling agents' utilities. Importantly, taxation does not have to be neither preferential, i.e., to be imposed only on PoW networks, nor permanent. By taxing transactions between (all) cryptocurrencies and fiat currencies, central authorities --- especially in countries like China where most mining rigs are concentrated --- can influence the transition of the growing PoW networks to the PoS technology (or any other desirable alternative with equivalent technological guarantees). Importantly, our work shows that it suffices to (i) exercise a temporary influence to the system and (ii) to only reach a critical mass of adopters of PoS that is well below the simple majority before the critical phase transition becomes permanent. \par
Finally, our auxiliary results concerning the functional relationship between the critical value, $T_c\(\gamma\)$, and the cost parameter, $\gamma$, cf. \Cref{cor:monotonicity}, suggest that temporary shocks in the electricity price would accelerate the effects of the transition mechanism. Specifically, since $T_c\(\gamma\)$ decreases in $\gamma$, an increase in the electricity price for PoW mining would sooner prompt the critical phase transition of the population to the alternative technology. These results, extrapolated to similar technological dilemmas --- CO$_2$ emitting vs electric (or generally environmentally friendly) vehicles, e.g., cars or aircrafts or sustainable urban planning --- provide an enhanced, and provably effective, toolbox in the hands of contemporary policymakers.

%% file: omitted_proofs.tex
\section{Appendix}\label{app:appendix}

\subsection{Additional Material and Omitted Proofs: \Cref{sec:model}}\label{app:omitted}

\begin{remark}[Derivation of revision protocol \eqref{eq:protocol}]For games with general action sets $A=\{1,2,\dots,n\}$, i.e., with possibly more than two strategies, the Q-learning agents update their strategies according to the following rule
\begin{alignat*}{3}
\max &&\quad \sum_{j=1}^nx_j&Q_t\(j\)-T\sum_{j=1}^nx_j\ln{x_j}\\
\text{subject to:} &&\quad \sum_{j=1}^nx_j&=1 \tag{S1'}\\
&&x_j&\ge 0, \qquad\text{for } j=1,2,\dots,n.
\end{alignat*}
In the current setting, at which each agent has two strategies, this conveniently reduces to the one-dimensional maximization problem in \Cref{eq:protocol}. To intuitively explain the objective function in (S1'), observe that the first term, i.e., $\sum_{j=1}^nx_j Q_t\(j\)$ enforces maximization of the $Q$-values. Since it is linear in the $x_j$'s, it would simply choose (put full probability on) the strategy with the highest $Q$-value, if the second term (entropy) was missing. However, the introduction of the entropy term, i.e., of $-\sum_{j=1}^nx_j\ln{x_j}$, essentially requires from the agent to choose the distribution $x$ with maximum entropy for every given weighted sum of the $Q$-values, and hence, to explore (assign positive probability to) sub-optimal strategies. \\
The relative importance between maximization of $Q$-values and exploration of the strategy space is controlled by parameter $T\ge0$. Termed \emph{temperature} in physics, $T$ can be interpreted as a tuning parameter: as $T\to0$, the agent always acts greedily and chooses the strategy corresponding to the maximum $Q$–value (pure exploitation), whereas as $T\to \infty$, the agent chooses a strategy completely at random (pure exploration). In particular, for $T=0$, the system reduces to the well-known replicator (best response) dynamics which (under standard regularity assumptions that are met in the present model) recover the Nash equilibria of the underlying evolutionary game, cf. \Cref{sub:elementary}. For different values of $T>0$, the resting points of the system change --- sometimes in an abrupt way in response to smooth changes in $T$ --- and this is precisely the intuition that we exploit here to design a mechanism that will stabilize the system and move it out of a lock-in to a desirable state. In fact, as shown in \cite{Yan17}, the temperature can be considered as a control parameter in the arsenal of a system designer. From the objective function, one discerns that parameter $T$ essentially rescales all the Q-values in a multiplicative way. Hence, as shown in \cite{Yan17}, $T$ can be treated as a taxation parameter in economic systems or as a medically controlled substance in health related settings. The interpretation of $T$ in a concrete application is further elaborated in \Cref{sec:applications}.
\end{remark}

\begin{proof}[Proof of \Cref{prop:evolutionary}]
The resulting game 

\begin{equation*}
P = \;\;\;\bordermatrix{
~ & W & S \cr
W & 1-\gamma & -\gamma \cr
S & 0 & 1 \cr}.
\end{equation*}
with $\gamma\in\(0,1\)$ has three Nash equilibria: two pure, $\(W,W\)$ with payoffs $\(1-\gamma,1-\gamma\)$ and $\(S,S\)$ with payoffs $\(1,1\)$, and one (fully) mixed $\(\(\frac{1+\gamma}2,\frac{1-\gamma}2\),\(\frac{1+\gamma}2,\frac{1-\gamma}2\)\)$ with payoffs $\(\frac{1-\gamma}2,\frac{1-\gamma}2\)$. These correspond to population states $x_1=0$, $x_2=\frac{1+\gamma}2$, and $x_3=1$. The corresponding payoffs can be also derived by substituting in the average payoff function, \eqref{eq:aggregate}, which here becomes
\[\bar{u}\(x\)= 2x^2-\(2+\gamma\)x+1.\]
All three Nash equilibria are symmetric. The $\(S,S\)$ (bottom right) Nash equilibrium is (strictly) payoff dominant, i.e., it \emph{Pareto-dominates} the other two (recall that $\gamma \in \(0,1\)$ is the cost from the current costly technology) and is also (strictly) \emph{risk} dominant, since $0+1>1-2\gamma$ for any $\gamma\in \(0,1\)$. Additionally, both pure strategy equilibria are evolutionary stable, since $u\(W,W\)=1-\gamma>0=u\(S,W\)$, and $u\(S,S\)=1>-\gamma=u\(W,S\)$. A symmetric mixed Nash equilibrium $\(x^*,x^*\)$ is evolutionary stable if $u\(x^*,x\)>u\(x,x\)$ for all other mixed strategies $x\neq x^*$. Hence, the mixed Nash equilibrium is \emph{not} evolutionary stable, since for any other $x\in\(0,1\)$ with $x\neq\frac{1+\gamma}2$
\begin{align*}
u\(\frac{1+\gamma}2,x\)-u\(x,x\)&=\(\frac{1+\gamma}{2}-x\)\lt \(1-\gamma\) x+\(-\gamma\)\(1-x\)\rt+\(\frac{1-\gamma}{2}-\(1-x\)\)\(1-x\)\\&=-\frac12\(2x-\(1+\gamma\)\)^2<0. 
\end{align*}
\end{proof}

\subsection{Omitted Proofs: \Cref{sec:results}}

\begin{proof}[Proof of \Cref{lem:explicit}]
From \eqref{eq:payoffs} and \eqref{eq:dynamics_pre}, it follows directly that
\begin{align*}
\dot x& =x\lt u\(W,x\)-\bar{u}\(x\)+T\(x\ln{\(\frac xx\)}+\(1-x\)\ln{\(\frac{1-x}{x}\)}\)\rt\\
&=x\lt \(1-x\)\lt u\(W,x\)-u\(S,x\)\rt+T\(1-x\)\ln{\(\frac{1-x}{x}\)}\rt\\
&=x\(1-x\)\lt  u\(W,x\)-u\(S,x\)+T\ln{\(\frac{1-x}{x}\)}\rt\\
&=x\(1-x\)\lt 2x-\(1+\gamma\)-T\ln{\(\frac{x}{1-x}\)}\rt.
\end{align*}
\end{proof}

\begin{lemma}\label{lem:critical}
Let $T\lt 0,1/2\rt$. Then, for any $\gamma\in\(0,1\)$, the equation
\begin{equation}\label{eq:t_critical}
\sqrt{1-2T}-\gamma-T\cdot\ln{\(\frac{1+\sqrt{1-2T}}{1-\sqrt{1-2T}}\)}=0
\end{equation}
has a unique solution $T_c\(\gamma\)$ or simply $T_c\in\(0,1/2\)$.
\end{lemma}

\begin{proof}[Proof of \Cref{lem:critical}]
Let $u:=\sqrt{1-2T}$. Then, $T\in\lt0,1/2\rt$ implies that $u\in\lt 0,1\)$ and the transformation is one to one with inverse $T=\frac12\(1-u^2\)$. Hence, we need to show that the function 
\[\g:=u-\gamma-\frac{1-u^2}{2}\ln{\(\frac{1+u}{1-u}\)}\]
has a unique root $u_c$ in $\(0,1\)$, so that $T_c\(\gamma\)=\frac12\(1-u_c^2\)$. Note that $\g$ is defined for any $u\in\(-1,1\)$. The derivative of $\g$ with respect to $u$ is 
\[\frac{d}{du}\g=u\ln{\(\frac{1+u}{1-u}\)}\ge0\]
for all $u\in\(-1,1\)$ with $g\(u\)=0$ only for $u=0$. Hence, $\g$ is strictly increasing with $\lim_{u\to -1^+}\g=-1+\gamma<0$, $g_\gamma\(0\)=0-\gamma<0$ and $\lim_{u\to1^-}\g=1-\gamma>0$. Accordingly $\g$ has precisely one root, $u_c\in\(0,1\)$ which yields the unique solution of the equation in the statement of Lemma \Cref{lem:critical}, given by $T_c\(\gamma\)=\frac12\(1-u_c^2\)\in\(0,1/2\)$.
\end{proof}

\begin{lemma}\label{lem:f}
Let $\gamma\in\(0,1\)$ and let $T_c\(\gamma\)\in\(0,1/2\)$ as given by \Cref{lem:critical}. Then, the number of the solutions of the equation $f\(x;T,\gamma\)=0$ with \[f\(x;T,\gamma\)=2x-\(1+\gamma\)-T\ln{\(\frac{x}{1-x}\)}\] depends on the relative value of $T>0$ to $T_c\(\gamma\)$ as follows
\begin{itemize}
\item $0<T<T_c\(\gamma\)$: there are 3 solutions $x_1,x_2,x_3$, with $x_1 \in\(0,x_l\)$, $x_2\in\(\(1+\gamma\)/2,x_u\)$ and $x_3\in\(x_u,1\)$, with $x_{u,l}$ as in \eqref{eq:notation}. Moreover, it holds that $0<x_l<1/2<\(1+\gamma\)/2<x_u<1$. 
\item $T=T_c\(\gamma\)$: there are 2 solutions $x_1,x_2$, with $x_1\in\(0,x_l\)$ and $x_2=x_u$, with $x_l<1/2$ and $x_u>\(1+\gamma\)/2$.
\item $T>T_c\(\gamma\)$: there is 1 solution $x_1$, with $x_1\in\(0,x_l\)$, with $x_l<1/2$ when $T< 1/2$ and $x_1\in \(0,1/2\)$ when $T\ge1/2$.
\end{itemize}
\end{lemma}

\begin{proof}[Proof of \Cref{lem:f}]
For any $T>0$ and $\gamma \in (0, 1)$, the function $f\(x; T, \gamma\)$ is continuous in $x\in\(0,1\)$ with
\begin{align}\label{eq:signs}
\lim_{x\to0^+}f\(x\)=+\infty,\qquad & \lim_{x\to1^-}f\(x;T,\gamma\)=-\infty,\nonumber\\
f\(1/2\)=-\gamma,\qquad & f\(\(1+\gamma\)/2\)=-T\ln{\(\frac{1+\gamma}{1-\gamma}\)<0}.
\end{align}
This implies that $f\(x;T,\gamma\)$ starts positive and ends up negative. The derivative of $f\(x;T,\gamma\)$ with respect to $x\in\(0,1\)$ is $f'\(x;T,\gamma\)=2-\frac{T}{x\(1-x\)}$, with $f'\(x;T,\gamma\)=0$, if and only if $-2x^2+2x-T=0$. If $T>1/2$, then $f'\(x;T,\gamma\)>0$ for any $x\in\(0,1\)$. For $T<1/2$, there exist two critical points 
\[x_{l,u}\(T\)=\frac12\(1\pm\sqrt{1-2T}\).\]
and for $T=1/2$ only one, i.e., $x_l=x_u=1/2$.
Depending on the value of $T$ relative to $T_c\(\gamma\)$, there are three cases for the number of the solutions of the equation $f\(x;T,\gamma\)=0$ for $x\in\(0,1\)$. 
\begin{itemize}
\item $0<T<T_c\(\gamma\)$. Since $T_c\(\gamma\)<1/2$ for any $\gamma\in\(0,1\)$ as shown in \Cref{lem:critical}, we have for any $T\in\(0,T_c\(\gamma\)\)$ by equation \eqref{eq:t_critical} that
\begin{align*}
x_u\(T\)&=\frac{1}{2}\(1+\sqrt{1-2T}\)>\frac{1}{2}\(1+\sqrt{1-2T_c\(\gamma\)}\)\\&
=\frac12\(1+\gamma+T_c\(\gamma\)\ln{\(\frac{1+\sqrt{1-2T_c\(\gamma\)}}{1-\sqrt{1-2T_c\(\gamma\)}}\)}\)>\frac12\(1+\gamma\).
\end{align*}
Moreover, it also immediate that $1-2T>0$ and, hence that $0<x_l<1/2$ which shows that $0<x_l<1/2<\(1+\gamma\)/<x_u<1$ as claimed. Also, by the strict monotonicity of $\g \in \(-1,1\)$ that was shown in the proof of \Cref{lem:critical}, we have that 
\begin{align*}
f\(x_l;T,\gamma\)&=g_\gamma\(-\sqrt{1-2T}\)<g_\gamma\(0\)<0,\\
f\(x_u;T,\gamma\)&=g_\gamma\(\sqrt{1-2T}\)>g_\gamma\(\sqrt{1-2T_c\(\gamma\)}\)=0.
\end{align*}
Hence, using equations \eqref{eq:signs} and the sign of $f'\(x;T,\gamma\)$, we obtain that $f\(x;T,\gamma\)$ has precisely one root $x_1\in\(0,x_l\)$, one root $x_2\in\(\(1+\gamma\)/2,x_u\)$ and one root $x_3\in\(x_u,1\)$.
\item $T=T_c\(\gamma\)$. As in the previous case, $T_c\(\gamma\)<1/2$ which by the same reasoning implies that $0<x_l<\(1+\gamma\)/2<x_u<1$. However, in this case, $f\(x_u;T,\gamma\)=g_\gamma\(\sqrt{1-2T_c\(\gamma\)}\)=0$ and hence, using again equations \eqref{eq:signs} and the sign of $f'\(x;T,\gamma\)$, it follows that $f$ has one root in $\(0,x_1\)$, and a second root in $\(\(1+\gamma\)/2,1\)$ which is precisely $x_u$.
\item $T>T_c\(\gamma\)$. In this case, 
\[f\(x_u;T,\gamma\)=g_\gamma\(\sqrt{1-2T}\)<g_\gamma\(\sqrt{1-2T_c\(\gamma\)}\)=0\]
which implies that $f\(x;T,\gamma\)$ turns negative at some point $x_1<x_l$ when $T<1/2$ or $x_1<1/2$ when $T>1/2$ (at which case $x_l$ does not exist) and remains negative thereafter since $x_2$ is a local maximum. Hence, this $x_1\in\(0,x_l\)$ or $x_1\in\(0,1/2\)$ is the unique root of $f\(x;T,\gamma\)$ in $\(0,1\)$. 
\end{itemize}
\end{proof}
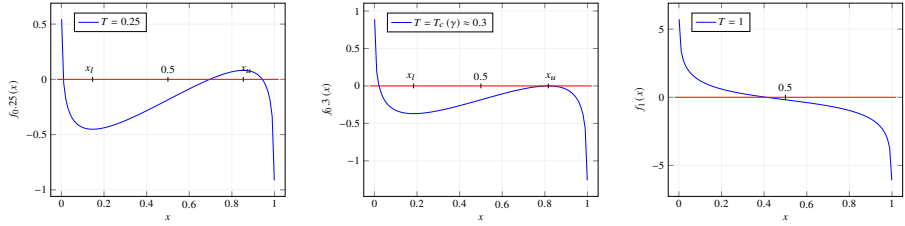
\begin{figure}[!htb]
\centering
\begin{tikzpicture}[scale=0.45]
\begin{axis}[xlabel=$x$, ylabel=$f_0.25\(x\)$, cycle list name=color list, xmajorgrids=true, ymajorgrids=true, grid style={line width=.1pt, draw=gray!10}, xmin=-0.05, xmax=1.05, legend entries = {$T=0.25$}, legend cell align={left}, legend style={at={(0.1,0.95)},anchor=north west}]
\addplot[blue,domain=0.001:0.999,samples=102]{2*x-1-0.185-0.25*ln(x/(1-x))};
\addplot[red,domain=-0.02:1.02, samples=10]{0}
[every node/.style={yshift=8pt}]
node[black,pos=0.1464]{$x_l$} 
node[pos=0.5,black] {$0.5$}
node[pos=0.8536,black] {$x_u$};
\addplot+[only marks,forget plot, mark=|] coordinates {(0.1464,0) (0.5,0) (0.8536,0)};
\end{axis}
\end{tikzpicture}\hspace{10pt}
\begin{tikzpicture}[scale=0.45]
\begin{axis}[xlabel=$x$, ylabel=$f_0.3\(x\)$, cycle list name=color list, xmajorgrids=true, ymajorgrids=true, grid style={line width=.1pt, draw=gray!10}, xmin=-0.05, xmax=1.05, legend entries = {$T=T_c\(\gamma\)\approx0.3$}, legend cell align={left}, legend style={at={(0.1,0.95)},anchor=north west}]
\addplot[blue,domain=0.001:0.999,samples=102]{2*x-1-0.185-0.3001*ln(x/(1-x))};
\addplot[red,domain=-0.02:1.02, samples=10]{0}
[every node/.style={yshift=8pt}]
node[black,pos=0.1838]{$x_l$} 
node[pos=0.5,black] {$0.5$}
node[pos=0.8162,black] {$x_u$};
\addplot+[only marks,forget plot, mark=|] coordinates {(0.1838,0) (0.5,0) (0.8162,0)};
\end{axis}
\end{tikzpicture}\hspace{10pt}
\begin{tikzpicture}[scale=0.45]
\begin{axis}[xlabel=$x$, ylabel=$f_1\(x\)$, cycle list name=color list, xmajorgrids=true, ymajorgrids=true, grid style={line width=.1pt, draw=gray!10}, xmin=-0.05, xmax=1.05, legend entries = {$T=1$}, legend cell align={left}, legend style={at={(0.1,0.95)},anchor=north west}]
\addplot[blue,domain=0.001:0.999,samples=102]{2*x-1-0.185-ln(x/(1-x))};
\addplot[red,domain=-0.02:1.02, samples=10]{0}
[every node/.style={yshift=8pt}]
node[pos=0.5,black] {$0.5$};
\addplot+[only marks,forget plot, mark=|] coordinates {(0.5,0)};
\end{axis}
\end{tikzpicture}
\caption{The function $f\(x;T,\gamma\)$ for $\gamma=0.185$ and $T=0.25, T=T_c\(\gamma\)\approx 0.3$ and $T=1$. For $T<T_c\(\gamma\)$, there are three steady states (roots of $f\(x;T,\gamma\)$), for $T=T_c\(\gamma\)$, there are precisely $2$, and for $T>T_c\(\gamma\)$, only one. The smallest root is always less than $1/2$ (in particular less than $x_l=\frac12\(1-\sqrt{1-2T}\)$, if $x_l$ exists), whereas the remaining ones (if any) are larger than $\(1+\gamma\)/2$.}
\label{fig:critical}
\end{figure}

The three cases of \Cref{lem:f} are illustrated in \Cref{fig:critical}. In the depicted instantiation, $\gamma=0.185$ which yields $T_c\(\gamma\)\approx0.3$. The three curves correspond to $T=0.25, T=T_c\(\gamma\)$ and $T=1$ which, in agreement with \Cref{lem:f}, yield $3,2$ and $1$ solutions respectively to the equation $f\(x;T,\gamma\)=0, x\in\(0,1\)$. 

Using \Cref{lem:critical,lem:f}, we can now determine the convergence properties of the Q-learning dynamical system $\dot x$ and hence, prove \Cref{thm:main}. 
\begin{proof}[Proof of \Cref{thm:main}]
The existence of the steady states in the three cases has been established in \Cref{lem:f}. Hence, it remains to prove the claims about their stability. The dynamics defined by $\dot x$ are 1-dimensional and hence their convergence properties and the stability of their steady states can be fully determined by the sign of $\dot x$. Since $x\(1-x\)>0$ for any $x\in\(0,1\)$, the sign of $\dot x$ fully depends on $f\(x;T,\gamma\)$. In turn, the sign of $f\(x;T,\gamma\)$ for any $x\in\(0,1\)$ has already been determined by the calculations in the proof of \Cref{lem:f} and in particular by equation \eqref{eq:signs}. Formally, we have the following cases
\begin{itemize}
\item $T<T_c\(\gamma\)$. In this case, the dynamics $\dot x$ have three steady states $x_1,x_2,x_3\in\(0,1\)$ that correspond to the respective roots of function $f\(x;T,\gamma\)$. In particular, it holds that $0<x_1<x_l<1/2<\(1+\gamma\)/2<x_2<x_u<x_3<1$ and the sign of $\dot x$ starts positive and alternates accordingly. This gives the stability results in \Cref{fig:stability_b}.
\item $T=T_c\(\gamma\)$. At this point, the two roots that are larger than $1/2$, namely $x_2$ and $x_3$, merge to one root precisely at $x_u$. The critical observation is that this new steady state is unstable, since the dynamics $\dot x$ have a negative sign at both sides of the root $x_u$. Moreover, it holds that $x_u>\(1+\gamma\)/2$. This is shown in \Cref{fig:stability_c}.
\item $T>T_c\(\gamma\)$. In this case, $f'\(x;T,\gamma\)<0$ which implies that the dynamics $\dot x$ are decreasing for any $x\(0,1\)$. Since $f\(x;T,\gamma\)$ starts positive and ends up negative, there remains only one root (steady state), $x_1$ of $f\(x;T,\gamma\)$, which is strictly less than $1/2$. An illustration is given in \Cref{fig:stability_d}.
\end{itemize}
\end{proof}
\begin{figure}[!htb]
\centering
\includegraphics[width=0.5\textwidth]{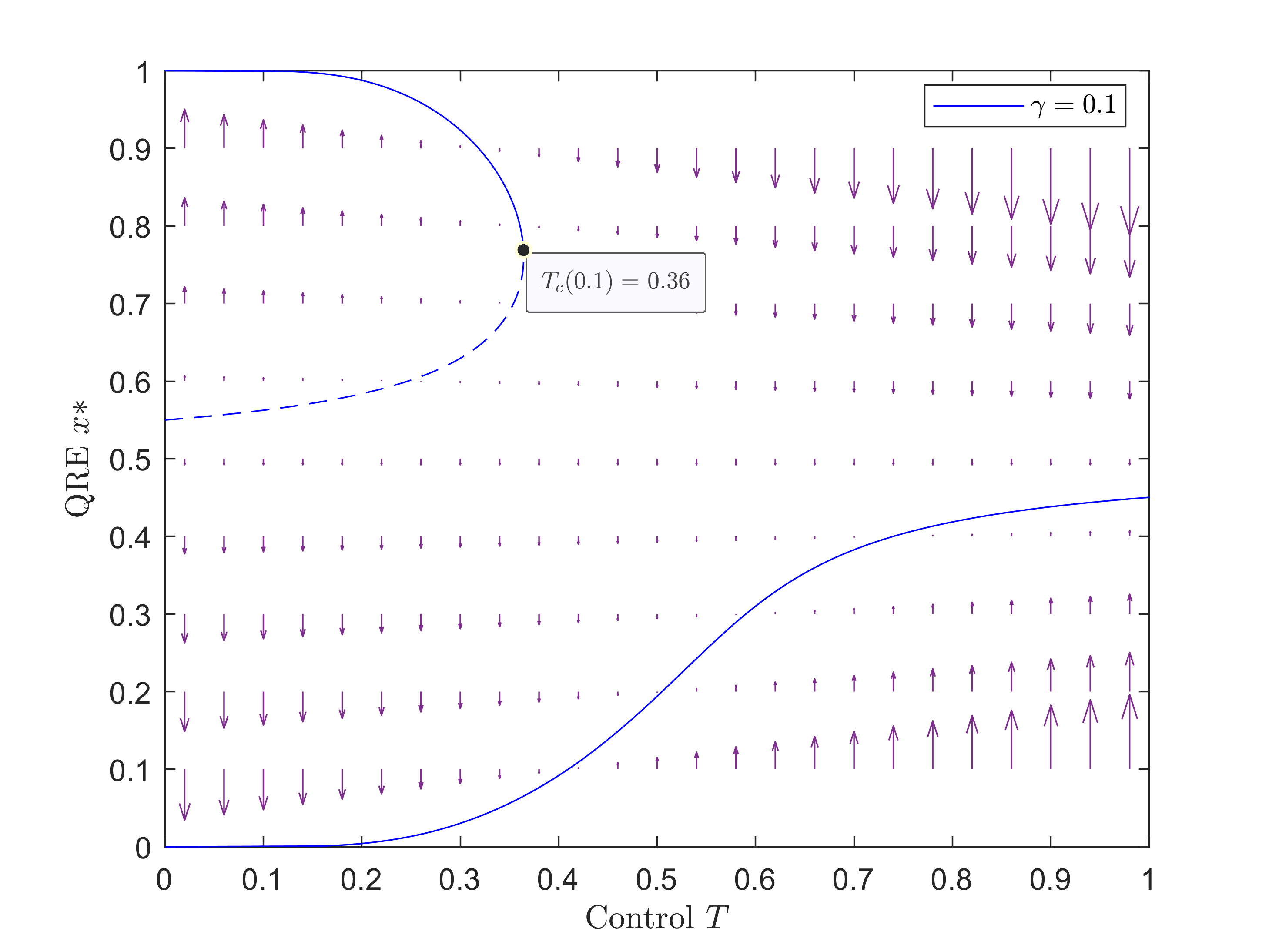}
\caption{An instantiation of the stability analysis of \Cref{thm:main} for $\gamma=0.1$ and all $T\in[0,1]$ (for values of $T\ge1$, the results are similar). The blue line represents the geometric locus of the QRE. Solid segments correspond to stable and dashed segments to unstable equilibria. This illustration can be compared to the stability analysis in \Cref{fig:stability}.}
\label{fig:stability_app}
\end{figure}
\begin{proof}[Proof of \Cref{cor:monotonicity}]
For $\gamma\in\(0,1\)$, the critical temperature $T_c\(\gamma\)$ is the unique solution of the equation 
\[\sqrt{1-2T}-\gamma-T\cdot\ln{\(\frac{1+\sqrt{1-2T}}{1-\sqrt{1-2T}}\)}=0,\]
in $(0,1/2)$, cf. \Cref{lem:critical}. To proceed, let $F\(\gamma,T\):=\sqrt{1-2T}-\gamma-T\cdot\ln{\(\frac{1+\sqrt{1-2T}}{1-\sqrt{1-2T}}\)}$. Implicit differentiation of the function $F\(\gamma,T\)=0$ with respect to $\gamma$, yields
\[0=\frac{\partial F\(\gamma,T\)}{\partial \gamma}+\frac{\partial F\(\gamma,T\)}{\partial T}\cdot \frac{d T}{d\gamma}=-1-\ln{\(\frac{1+\sqrt{1-2T}}{1-\sqrt{1-2T}}\)}\cdot\frac{d T}{d\gamma}\]
Hence 
\[\frac{d T}{d\gamma}=-\(\ln{\(\frac{1+\sqrt{1-2T}}{1-\sqrt{1-2T}}\)}\)^{-1} \]
which is negative, since the argument of the $\ln$ is larger than $1$ for all $T\in (0,1/2]$.
\end{proof}

\subsection{Omitted Materials: \Cref{sub:hysteresis}}

\begin{proof}[Proof of \Cref{cor:convergence}]
As mentioned in \Cref{rem:sequence}, in theory it suffices to select $\langle T_0=0, T_1>T_c\(\gamma\), T_2=0\rangle$. Then, the stability properties of the dynamics that were established in \Cref{thm:main} directly imply the result. For practical purposes though, only gradual changes in $T$ may be possible. However, we can show that the mechanism still results in the same outcome. Specifically, starting from any initial point $x_0\in[\(1+\gamma\)/2,1]$ and increasing $T$ from $T_0=0$ to a $T_1\in\(0,T_c\(\gamma\)\)$, the dynamics will stabilize at QRE $x_3$. Further increasing $T$ above $T_c\(\gamma\)$ will lead the dynamics to stabilize at the unique remaining QRE, $x_1^*\(T\)<1/2$. It remains to show that we can reduce $T$ back to $0$ and converge to the desired equilibrium $x=0$. By \Cref{lem:f}, the function $f\(x;T,\gamma\)=2x-1-\gamma-T\ln{\(\frac{x}{1-x}\)}$ has precisely one root $x_1\(T\)$ in $\(0,x_l\(T\)\)$ for any value of $0<T\le T_c\(\gamma\)<1/2$, where $x_l\(T\)=\frac12\(1-\sqrt{1-2T}\)$. Moreover, implicit differentiation of the function $f\(x,T\)=0$, with $f\(x,T\):=2x\(T\)-1-\gamma-T\ln{\(\frac{x\(T\)}{1-x\(T\)}\)}$, shows that $x\(T\)$ is strictly increasing in $T$ for $x<1/2$, since 
\[\frac{\partial f\(x,T\)}{\partial T}\frac{d x}{d T}+\frac{\partial f\(x,T\)}{\partial x}\frac{d x}{d T}=-\ln{\(\frac{x}{1-x}\)}+\(2-\frac{x}{1-x}\)\frac{d x}{d T}=0\]
which yields 
\[\frac{d x}{d T}=\frac{\ln{\(\frac{x}{1-x}\)}}{2-\frac{T}{x\(1-x\)}}>0\]
for $x<x_l\(T\)$ and $T<T_c\(\gamma\)<1/2$. Hence, $0<x_1\(T\)<x_l\(T\)$ implies that for any $\epsilon>0$, there exists a $\delta>0$ such that $x_1\(T\)<\epsilon$ for any $T<\delta$, since
\[\lim_{T\to0^+}x_l=\lim_{T\to0^+}\frac12\(1-\sqrt{1-2T}\)=0\]
This implies, that $x_1\(T\)\to 0$ as $T\to 0^+$, which concludes the proof. 
\end{proof}

An alternative visualization of the control catastrophe mechanism described in \Cref{sub:hysteresis} is provided in \Cref{app:figure}.
\begin{figure}[!htb]
\centering
\begin{minipage}[t]{0.499\textwidth}
\centering
\includegraphics[width=\textwidth,trim=0.9cm 5.5cm 1cm 0.3cm, clip]{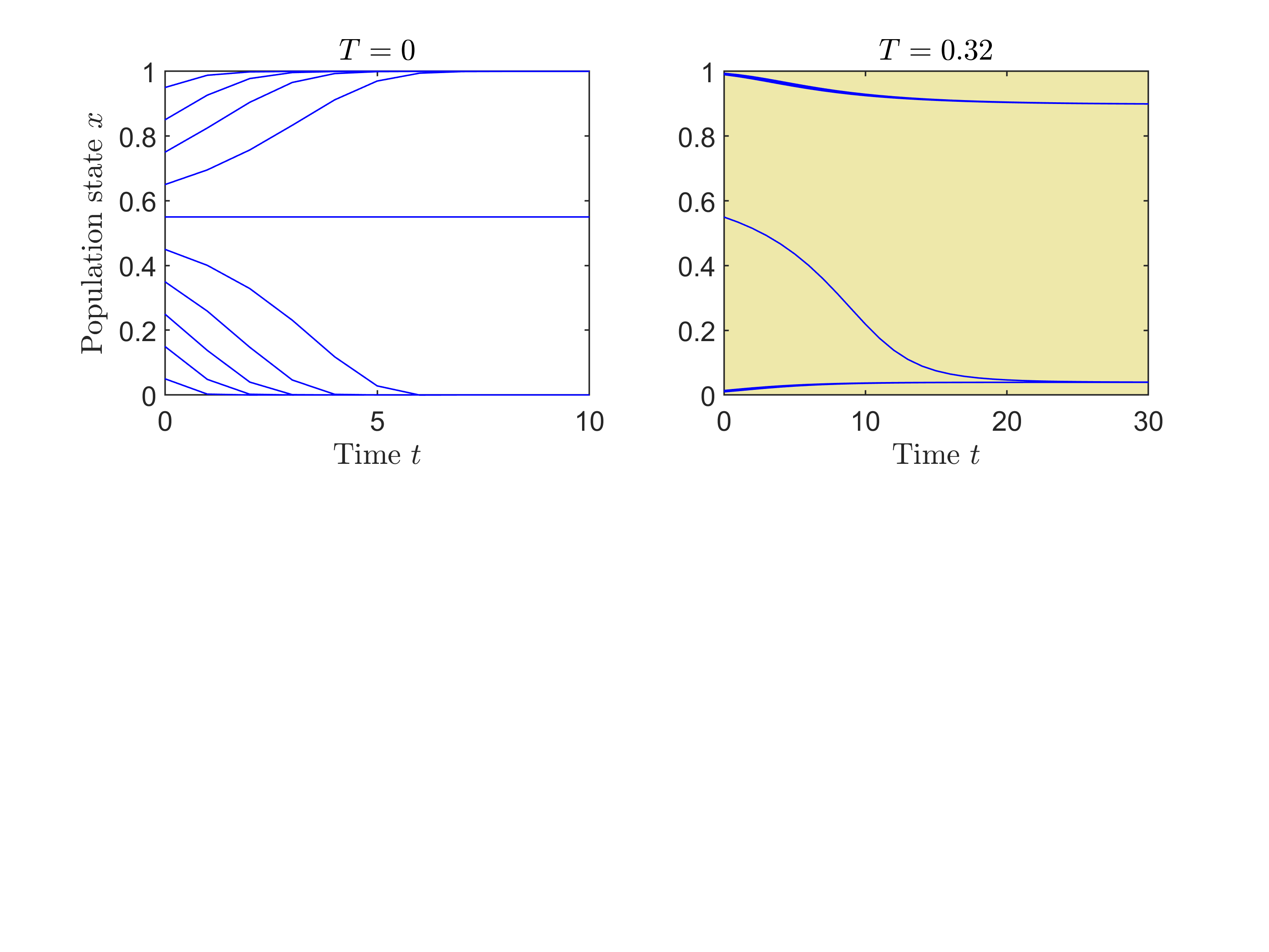}
\end{minipage}\hfill
\begin{minipage}[t]{0.499\textwidth}
\centering
\includegraphics[width=\textwidth,trim=0.9cm 5.5cm 1cm 0.3cm, clip]{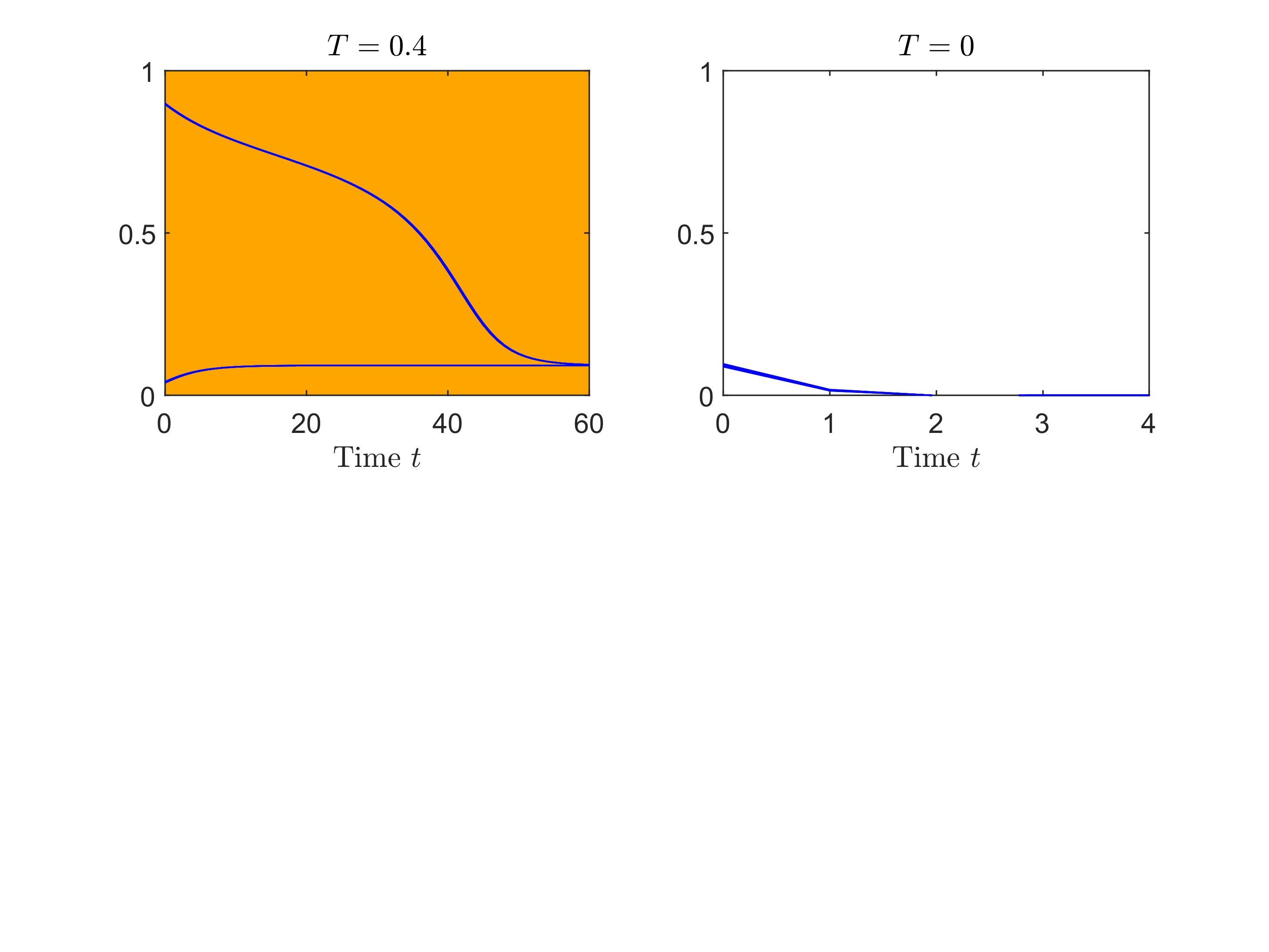}
\end{minipage}
\caption{The process of destabilizing an equilibrium via a controlled catastrophe and hysteresis mechanism in the game \eqref{eq:game}. The first panel shows the evolution of the population state for 10 initial starting points and control parameter $T=0$. The Q-learning dynamics -- in this case, equivalent to the replicator dynamics -- converge to the two stable equilibria of the system. The only trajectory that converges to the mixed, unstable equilibrium is the one that precisely starts from it (horizontal line at $\(1+\gamma\)/2$, i.e., slightly above $1/2$). In the second panel, the control parameter increases (but remains below the critical level). The starting points are the possible steady states for $T=0$. The two stable equilibria (QRE) are in the interior of the admissible region, $x\in\(0,1\)$. There is a third unstable equilibrium that separates the attracting regions of the two equilibria (not visible since no trajectory converges to it). In the third panel, the control parameter is increased above the critical level. Starting again from the endpoints of the previous panel, the stability of the system changes abruptly: the upper and middle equilibria merge and annihilate each other. There is only one remaining equilibrium and the population converges to it independently of the starting point. This equilibrium lies in the attracting region of the efficient technology (risk-dominant equilibrium), which can be now recovered by resetting the system back to its initial control value, $T=0$ (last panel).}
\label{app:figure}
\end{figure}

\section{Appendix}
\label{app:b}

\subsection{Omitted Materials and Proofs: \Cref{sub:robust_e}}
\label{app:section6}

\begin{lemma}\label{lem:bound}
Let $f\(x;\gamma,T\)$ as in \Cref{def:f} and let $f_\epsilon\(x;T,\gamma\):=f\(x;T,\gamma\)+\epsilon\(x\)$ where $\epsilon\(x\)$ denotes a noise term defined for $x\(0,1\)$ such that $|\epsilon\(x\)|\le \epsilon_0$ for some $\epsilon_0\in \(0,\min{\{\gamma,1-\gamma\}}\)$. Then, it holds that 
\[f\(x;T,\gamma+\epsilon_0\)< f_\epsilon\(x;T,\gamma\)<f\(x;T,\gamma-\epsilon_0\)\]
\end{lemma}
\begin{proof}
Since $\frac{\partial}{\partial \gamma} f\(x;T,\gamma\)=-1<0$, $f\(x;T,\gamma\)$ is decreasing in $\gamma$ which implies that 
\[f\(x;T,\gamma+\epsilon_0\)<f\(x;T,\gamma\)+\epsilon\(x\)=f_\epsilon\(x;T,\gamma\)<f\(x;T,\gamma-\epsilon_0\).\]
\end{proof}

\begin{notation}\label{not:min}
For given $\gamma\in\(0,1\)$ and any $T>0$, let $x_m^*\(T\):=\min{\{x\in\(0,1\): f\(x;T,\gamma\) =0\}}$ and $x_m^*\(T;\epsilon_0\):=\min{\{x\in\(0,1\): f_\epsilon\(x;T,\gamma\) =0\}}$ denote the minimum QRE at level $T$ of the unperturbed and perturbed dynamics, respectively. Note that by \Cref{thm:main}, $x^*\(T;\gamma\)<1/2$ for any pair $\(T,\gamma\)$.
\end{notation}

\begin{proof}[Proof of \Cref{thm:stability_e}]
The stability analysis of the dynamics in the perturbed case is derived in the same way as in \Cref{thm:main} and follows directly from \Cref{lem:bound} and the fact that $|\epsilon\(x\)|\le \epsilon_0$. In particular, let $x^*\in\(0,1\)$ be a QRE of the unperturbed system. Then, since $|\epsilon\(x\)|\le \epsilon_0$, we can found a $\delta\(x^*\)>0$, so that $|f_\epsilon\(x^*\pm\delta\(x^*\);\gamma,T\)|>\epsilon_0$. Hence, outside the $\delta\(x^*\)$-neighborhood of $x^*$, the function $f_\epsilon\(x;\gamma,T\)$ retains the sign of the unperturbed function $f\(x;T,\gamma\)$ and the stability property of the dynamics directly follow from \Cref{thm:main}.\par
To obtain an estimate of $\delta\(x^*\)$, observe that for any $x\in\(0,1\)$, by first order Taylor's theorem, we have that
\[f\(x^*\pm\delta;T,\gamma\)=f\(x^*;T,\gamma\)\pm\(2-T\frac{1}{\xi\(1-\xi\)}\)\delta=\pm\(2-T\frac{1}{\xi\(1-\xi\)}\)\delta,\]
where $\xi \in \(x^*-\delta,x^*\)$ or $\xi \in \(x^*,x^*+\delta\)$, respectively and $\delta>0$. This equation can then be solved to find a minimum value $\delta\(x^*\)$ so that $|f\(x^*\pm\delta\(x^*\);T,\gamma\)\ge\epsilon_0$. As claimed in \Cref{sec:robustness}, the steep increase, $\frac{1}{x\(1-x\)}$, of the entropy term as $x^*$ approaches the boundary, dominates the (bounded) perturbation term $\epsilon\(x^*\)$ and yields lower values of $\delta\(x^*\)$. This implies that the distance between $x_m^*\(T\)$ and $x_m^*\(T;\epsilon_0\)$ as defined in \Cref{not:min} is decreasing for values of $x^*_m\(T\)$ close to $0$. Thus, convergence to $x=0$ (which corresponds to the efficient technology) is still achieved. 
\end{proof}
\subsection{Omitted Materials and Proofs: \Cref{sub:robust_a}}

\begin{remark}\label{rem:alpha}
Parameter $\alpha$ expresses the value that is created by the technology in response to its adoption. In particular, there are three interesting cases, depending on whether $\alpha$ is smaller, larger than or equal to $1$.
\begin{itemize}
\item $\alpha<1$: Subadditive value. In this case, the total value generated from the two technologies is subadditive implying that a split of the network is more beneficial for the society. 
\item $\alpha=1$: Linear value. In this case, the aggregate value that is generated is increasing in the rate of adoption of the innovative (less costly) technology. Resolution of this case is trivial from a mathematical perspective, since the less costly technology constitutes a strictly dominant strategy for the population for any $\gamma\in\(0,1\)$.
\item $\alpha>1$: Superadditive value. In this case, the aggregate value is (locally) maximized when either technology is fully adopted. This case corresponds to direct (positive) network effects and is the one that is discussed in the current paper.
\end{itemize}
Typical instantiations of these cases are depicted in the two panels of \Cref{fig:alpha_1}.
\end{remark}
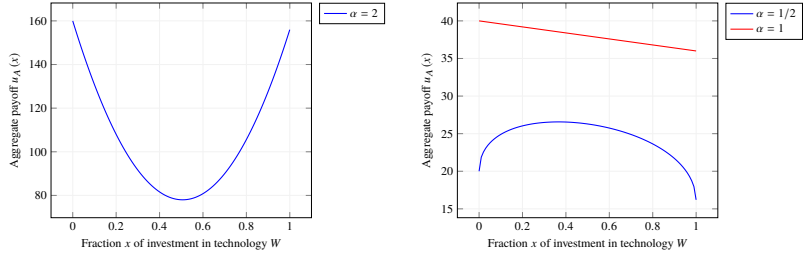
\begin{figure}[!htb]
\centering
\begin{tikzpicture}[scale=0.5]
\begin{axis}[xlabel=Fraction $x$ of investment in technology $W$, ylabel=Aggregate payoff $u_A\(x\)$, cycle list name=color list, xmajorgrids=true, ymajorgrids=true, grid style={line width=.1pt, draw=gray!10}, legend entries = {$\alpha=2$}, legend cell align={left}, legend pos={outer north east}]
\addplot[thick,blue,domain=0:1, samples=100,unbounded coords=jump]{10*(4^2)*(x^2-(1/(10*4^(2-1)))*x+(1-x)^2)};
\end{axis}
\end{tikzpicture}\hspace{10pt}
\begin{tikzpicture}[scale=0.5]
\begin{axis}[xlabel=Fraction $x$ of investment in technology $W$, ylabel=Aggregate payoff $u_A\(x\)$, cycle list name=color list, xmajorgrids=true, ymajorgrids=true, grid style={line width=.1pt, draw=gray!10}, legend entries = {$\alpha=1/2$, $\alpha=1$}, legend cell align={left}, legend pos={outer north east}]
\addplot[thick, blue, domain=0:1, samples=100,unbounded coords=jump]{20*(x^(1/2)-x/5+(1-x)^(1/2))};
\addplot[thick, red, domain=0:1, samples=100,unbounded coords=jump]{10*4*(x-(1/(10))*x+1-x)};
\end{axis}
\end{tikzpicture}
\caption{Aggregate payoff $u_A\(x\)$ created by the population as a whole at population state $x\in[0,1]$ (investment in technology $W$) for various values of parameter $\alpha$. The aggregate payoff (vertical axis) is calculated by the formula: $u_A\(x\)=VK^\alpha\lt x^\alpha-\frac{\gamma}{VK^{\alpha-1}}x+\(1-x\)^\alpha\rt$, with $x\in[0,1]$ and selected values of the parameters $V=10, K=4$ and $\gamma=1$. For $\alpha=2$ (left panel) and in general for $\alpha>1$, the total aggregate value is (locally) maximized at the boundaries, i.e., when either technology is fully adopted (superadditive value). The global maximum is attained when the less costly technology $S$ is fully adopted, i.e., when $x=0$. By contrast, for $\alpha=1/2$ (right panel, blue line), and in general for $\alpha<1$, the aggregate wealth is maximized when the population is split between the two technologies (subadditive value). For $\alpha=1$ (right panel, red line), the aggregate wealth is increasing in the adoption of the less costly technology $S$ (linear value).}
\label{fig:alpha_1}
\end{figure}

The proof of \Cref{thm:unique} relies on \Cref{lem:ln} which exploits a particular inequality of the natural logarithm.

\begin{lemma}\label{lem:ln}
For any $\alpha\ge2$, and any $x\in[0,1]$, it holds that 
\[x^\alpha\(1-x\)+x\(1-x\)^\alpha\le\frac{1}{2\alpha}.\]
\end{lemma}
\begin{proof}
Since $x\in[0,1]$, we can rewrite the inequality as 
\[x^\alpha\(1-x\)+x\(1-x\)^\alpha\le\frac{x}{2\alpha}+\frac{1-x}{2\alpha}\]
which is equivalent to 
\[\lt x^\alpha\(1-x\)-\frac{x}{2\alpha}\rt+\lt x\(1-x\)^\alpha-\frac{1-x}{2\alpha}\rt\le 0\]
Hence, by symmetry, it suffices to show that $x^\alpha\(1-x\)-\frac{x}{2\alpha}\le0$, for any $x\in[0,1]$ and $\alpha>2$, which is in turn equivalent to $x^{\alpha-1}\(1-x\)\le\frac{1}{2\alpha}$. By differentiating the left hand side with respect to $x$, we find that 
$\frac{d}{dx}\(x^{\alpha-1}\(1-x\)\)=\frac{x^{\alpha-2}}{\alpha}\(\frac{\alpha-1}{\alpha}-x\)$ 
which is zero for $x=\frac{\alpha-1}{\alpha}$. Since $x^{\alpha-1}\(1-x\)$ is equal to $0$ for both $x=0$ and $x=1$, it attains a maximum at $x=\frac{\alpha-1}{\alpha}$ with value $\(\frac{\alpha-1}{\alpha}\)^{\alpha-1}\frac{1}{\alpha}$. Accordingly, it suffices to show that 
\[\(\frac{\alpha-1}{\alpha}\)^{\alpha-1}\frac{1}{\alpha}\le \frac{1}{2\alpha}\]
or equivalently that $\(\frac{\alpha-1}{\alpha}\)^{\alpha-1}\le 1/2$ for any $\alpha\ge2$. However, the term on the left side is decreasing in $\alpha$, since by taking the logarithm and applying the inequality $\ln{\(x\)}\le x-1$, we obtain that
\begin{align*}
\frac{d}{d\alpha}\ln{\(\frac{\alpha-1}{\alpha}\)^{\alpha-1}}&=\ln{\(\frac{\alpha-1}{\alpha}\)}+\frac{1}{\alpha}\le \frac{\alpha-1}{\alpha}-1+\frac1\alpha=0.
\end{align*}
Hence, the maximum of the left side is attained for $\alpha=2$, yielding a value of $\(\frac{2-1}{2}\)^{2-1}=\frac12$, which concludes the proof.
\end{proof}

To ease the proof of \Cref{thm:unique}, we restrict attention to integer $\alpha>1$, yet the results continue to hold for any $\alpha\ge 1$. Recall that for $\alpha=1$, the efficient technology constitutes a strictly dominant strategy and hence its adoption is trivial (and hence, it is excluded from the statement of \Cref{thm:unique}). The continuous case requires an additional technical step for the interval $(2,3)$ and since it does not add much intuition, it is omitted. In particular, for $\alpha\in [2,3]$, $f_\alpha\(x\)$ is not monotone decreasing. However the statement of \Cref{thm:unique} continues to hold in this interval as well. Since this case requires more technical details without providing any additional insight, its proof is omitted.\par

\begin{proof}[Proof of \Cref{thm:unique}] 
The case $\alpha=2$ has been treated in the \Cref{thm:main} and \Cref{cor:convergence}. For $\alpha \ge 3$, $\gamma \in \(0, 1\)$ and $T \ge \frac{1}{2}$, the function $f_\alpha\(x\):= x^{\alpha - 1} - \(1 - x\)^{\alpha - 1} - \gamma - T\ln\frac{x}{1 - x}$ is continuous and satisfies
\begin{align*}
\lim_{x \to 0^+}f_\alpha(x) &= \lim_{x\to 0^+}\(x^{\alpha - 1} - (1 - x)^{\alpha - 1} - \gamma - T\ln\frac{x}{1 - x} \) = \infty \quad \text{and }\quad f_\alpha\(\frac{1}{2}\)= - \gamma < 0.
\end{align*}
Hence, since $f_\alpha\(x\)$ is continuous for $x\in(0, 1)$, there exists $x^* \in (0, \frac{1}{2})$ such that $f_\alpha\(x^*\)=0$. To prove uniqueness, it will be sufficient to prove that for $T\ge1/2$, $f_\alpha\(x\)$ is decreasing in $x$. This implies that for $T\ge1/2$, there is a unique steady state $x^*$ and hence, the critical temperature $T_c$ will necessarily satisfy $T_c< 1/2$. Taking the derivative of $f_\alpha\(x\)$ with respect to $x$, we obtain
\[\frac{d}{dx}f_\alpha\(x\)=\(\alpha-1\)x^{\alpha-2}+\(\alpha-1\)\(1-x\)^{\alpha-2}-T\frac{1}{x\(1-x\)}.\]
The last expression is decreasing in $T$ and hence it suffices to prove that $\frac{d}{dx}f_\alpha\(x\)\le 0$, for $T=1/2$. In this case, $\frac{d}{dx}f_\alpha\(x\)\le0$ is equivalent to 
\[x^{\alpha-1}\(1-x\)+x\(1-x\)^{\alpha-1}\le\frac{1}{2\(\alpha-1\)}\]
and the claim follows from \Cref{lem:ln}. In particular, equality holds only if $T = \frac{1}{2}$, $\alpha = 3$ ($\alpha - 1 = 2$), and $x = \frac{1}{2}$. Hence, $f_\alpha\(x\)$ is decreasing in $(0, 1)$ which proves the claim. Finally, we need to show that for $T=0$ and any $\alpha>1$, the system still has 3 steady states, $x_1=0,x_3=1$ and $x_2\in\(1/2,1\)$. For $T=0$, the first derivative of $f_\alpha\(x\)$ with respect to $x$ becomes
\[\frac{d}{dx}f_\alpha\(x\)=\(\alpha-1\)\lt x^{\alpha-2}+\(1-x\)^{\alpha-2}\rt>0\]
for all $x\in\(0,1\)$ and all $\alpha>1$ which implies that $f_\alpha\(x\)$ is monotonically increasing. Since $f_\alpha\(x\)$ starts negative at $x=0$, ends up positive at $x=1$, and $f_\alpha\(1/2\)=-\gamma<0$, it has one root that lies in $(1/2,1)$. Hence, the dynamics 
\[\dot x=x\(1-x\)\lt x^{\alpha-1}-\gamma-\(1-x\)^{\alpha-1}\rt\] have two obvious steady states $x_1=0$ and $x_3=1$ and a third steady state $x_2\in\(1/2,1\)$. Accordingly, for $T=0$, the usual stability analysis applies, which proves that $[0,1/2]$ lies in the attracting region of $x=0$ as claimed.
\end{proof}